\documentclass[numbib]{cta-author}
\usepackage{amsthm,amssymb,amsmath}
\usepackage{graphicx}

\usepackage{subfigure}
\usepackage{multirow, multicol}
\usepackage[english]{babel}

\usepackage{tikz}
\usetikzlibrary{shapes,arrows,calc,positioning,arrows.meta}

\usepackage{nomencl}
\makenomenclature
\makeindex

\usepackage{cleveref}
\usepackage{upgreek}
\usepackage{mathrsfs}
\usepackage{upgreek}
\usepackage{mathtools, cuted}
\usepackage{flushend}

\usepackage{algorithm}
\usepackage{algpseudocode}

\newcommand{\beq}{\begin{equation}}
\newcommand{\eeq}{\end{equation}}
\newcommand{\benn}{\begin{equation*}}
\newcommand{\eenn}{\end{equation*}}

\newtheorem{theorem}{Theorem}{}
{}
{}

\begin{document}

\tikzset{
block/.style = {draw, fill=white!20, rectangle, thin, minimum height=2.0em, minimum width=1.50em},
tmp/.style  = {coordinate}, 
sum/.style= {draw, fill=white!20, circle, node distance=.5cm},
tip/.style = {->, >=stealth', thin, dashed},
input/.style = {coordinate},
output/.style= {coordinate},
pinstyle/.style = {pin edge={to-,thin,black}},
startstop/.style = {draw, rounded rectangle, text centered, draw=black,thick},
io/.style = {trapezium, trapezium left angle=70, trapezium right angle=110, text centered},
process/.style = {rectangle, text centered, draw=black,thick},
decision/.style = {diamond, text centered, draw=black,thick},
arrow/.style = {-{Stealth[scale=1.2]},rounded corners,thick}
}


\title{Robust delay-dependent LPV output-feedback blood pressure control with real-time Bayesian estimation}

\author{\au{S. Tasoujian} \au{S. Salavati} \au{M. Franchek} \au{K. Grigoriadis}}

\address{{Department of Mechanical Engineering, University of Houston, Houston, TX, USA, 77204}
\email{stasoujian@uh.edu}}

\begin{abstract}
\looseness=-1 Mean arterial blood pressure (MAP) dynamics estimation and its automated regulation could benefit the clinical and emergency resuscitation of critical patients. In order to address the variability and complexity of the MAP response of a patient to vasoactive drug infusion, a parameter-varying model with a varying time delay is considered to describe the MAP dynamics in response to drugs. The estimation of the varying parameters and the delay is performed via a Bayesian-based multiple-model square root cubature Kalman filtering approach. The estimation results validate the effectiveness of the proposed random-walk dynamics identification method using collected animal experiment data. Following the estimation algorithm, an automated drug delivery scheme to regulate the MAP response of the patient is carried out via time-delay linear parameter-varying (LPV) control techniques. In this regard, an LPV gain-scheduled output-feedback controller is designed to meet the MAP response requirements of tracking a desired reference MAP target and guarantee robustness against norm-bounded uncertainties and disturbances. In this context, parameter-dependent Lyapunov-Krasovskii functionals are used to derive sufficient conditions for the robust stabilization of a general LPV system with an arbitrarily varying time delay and the results are provided in a convex linear matrix inequality (LMI) constraint framework. Finally, to evaluate the performance of the proposed MAP regulation approach, closed-loop simulations are conducted and the results confirm the effectiveness of the proposed control method against various simulated clinical scenarios.
\end{abstract}

\maketitle

\section{Introduction}\label{sec:Introl}


The human body has inherent feedback loops to maintain homeostasis including the regulation of blood pressure that may fail to work properly under severe trauma or disease or due to the administration of certain drugs. For this purpose, mean arterial blood pressure (MAP) regulation of a patient to a desired target value is essential in many clinical and operative procedures in critical care, and has been a challenging aspect of emergency resuscitation. Mainly, two types of vasoactive drugs are being used to attain a target MAP in emergency resuscitation: (1) vasodilator drugs to decrease the MAP to a target value, like sodium nitroprusside (SNP) which reduces the tension in the blood vessel walls \cite{he1986multiple}, and (2) vasopressor drugs to increase the MAP to a target value, like phenylephrine (PHP) which stimulates the depressed cardiovascular system causing vasoconstriction \cite{neves2010phenylephrine}.

Typically, MAP control and regulation procedures in clinical care are carried out manually using a syringe or infusion pump with a manual titration by the medical personnel. In these cases, drug delivery and adjustment may not be precisely managed, which can lead to undesirable or potentially fatal consequences, such as, increased cardiac workload and cardiac arrest. Moreover, manual drug administration is a time-consuming and labor-intensive task and often is challenged by poor and sluggish performance. Further, inaccurate operator monitoring can lead to under- or over-resuscitation with potentially dangerous outcomes \cite{luspay2016adaptive2, kee2005prevention}. Accordingly, the automation of the vasoactive drug infusion via feedback control has been proposed as a potential remedy to tackle the mentioned challenges of manual drug administration \cite{bailey2005drug}. To address the automated MAP regulation problem, several approaches including fractional-order  proportional-integral (PI) control \cite{sondhi2015fractional}, \textcolor{black}{nonlinear proportional-integral-derivative (PID) digital control \cite{slate1982automatic}}, adaptive predictive control \cite{kashihara2004adaptive, hahn2002adaptive}, \textcolor{black}{robust multiple-model adaptive control \cite{malagutti2013robust},} switching robust control \cite{ahmed2016design},  reinforcement learning \cite{sandu2016reinforcement}, and more recently PID and loop-shaping control methods \cite{StasoujianRobust} have been considered.

Among the automated MAP control strategies, model-based approaches have the advantage of fast, accurate, and reliable drug administration in the face of model mismatch, disturbances and noise.  However, the main challenge is due to the considerable intra- and inter-patient variations in the physiological MAP response to the drug infusion implying model parameters variation over time for an individual, as well as, from patient-to-patient \cite{kashihara2004adaptive}. Therefore, due to such physiological and pharmacological variations, a mathematical model with fixed parameters is inadequate to capture an individual's MAP response dynamics. In this regard, in order to improve the automated closed-loop resuscitation strategies, parameter-varying blood pressure response modeling and real-time estimation of the model's time-varying parameters is of significant practical interest. On this basis, in the present study, a first-order model with a time-varying delay and time-varying gain and time constant is considered to characterize the MAP response to the infusion of the vasopressor drugs used to regulate blood pressure in critical hypotensive scenarios. 

  Traditional parameter estimation methods, such as the recursive least-squares algorithm and instrumental-variable methods have been examined for real-time parameter estimation \cite{arnsparger1983adaptive, rao2003experimental, ljung1983theory}. Specifically, variance models have been proposed to characterize the MAP response of patients to drug infusion. However, these methods fail to sufficiently address the pharmacological variability problem and often suffer from a slow convergence rate \cite{bailey2005drug, Craig2004}. In more recent work, \cite{luspay2014design} utilizes the extended Kalman filtering (EKF) method for the real-time parameter estimation of a MAP response model. Although this approach can provide real-time parameter identification of a patient's MAP response model,  the estimation can be inaccurate when the response is far away from the equilibrium point, since the EKF is based on local linearization \cite{simon2010kalman}. Moreover, the proposed parameter identification approach is not capable of providing a consistent estimate of the time-lag parameter of the first-order mathematical model. Thus, to overcome the various inherent limitations of the previously utilized estimation methods, in this work, we develop a multiple-model square-root cubature Kalman filter (MMSRCKF) as a novel real-time model parameter and time-delay estimation method of the MAP response dynamics. MMSRCKF is a Bayesian filtering approach that can provide precise estimation of the varying model parameters and addresses the stochasticity in the nonlinear model without a need for linearization. Additionally, a linear parameter-varying (LPV) gain-scheduling controller combined with the real-time model parameter estimation is proposed to enable automated closed-loop drug delivery to meet the MAP regulation objectives in critical patient resuscitation. 
  
Automated MAP regulation should be robust against physiological disturbances and be able to adapt to varying patient dynamics. The varying MAP response dynamics and the large input time-delay degrade the performance of the closed-loop system by affecting its damping characteristics and bandwidth. Time-domain methods based on Lyapunov-Krasovskii functionals and Lyapunov-Razumikhin functions, to assess the stability of linear time-invariant (LTI) time-delay systems have been examined in \cite{Fridman2014, Wu2010}.  Control of time-delay LPV systems has been studied in \cite{Mohammadpour2012, Briat2015, salavati2019reciprocal, tasoujian2017parameter}. The corresponding stability criteria fall into delay-dependent and delay-independent sufficient conditions where the former criterion is generally considered to be less conservative. Mean-square stability of stochastic LPV systems with delayed measurements has been studied in \cite{Zhang20142}. \textcolor{black}{The authors in \cite{wang2007gain}, derived delay-dependent sufficient conditions for the closed-loop stabilization of LPV systems with input delay. A transformaton based on the maximum value of the delay is used to recast the original system into a more tractable form. A gain-scheduled static state-feedback controller is then designed to meet the performance requirements.} In another work, a robust static gain-scheduled controller design for discrete-time polytopic LPV systems with a state delay is formulated in a delay-independent matrix inequality framework in \cite{Rosa2018}. Dilated delay-dependent linear matrix inequalities (LMIs) for the control of state-delay polytopic LPV systems has been addressed in \cite{Nejem2018}. Through this method, the coupling between controller matrices and Lyapunov matrix functions is avoided and a gain-scheduled dynamic output feedback controller with memory is designed to reject disturbances. For the LPV MAP response control problem, \cite{Luspay2015} proposed an LPV control framework which uses Pad\'{e} approximation to transform the infinite-dimensional time-delay model into a non-minimum phase rational transfer function. The dynamics of the MAP response is assumed to be fully known; however, parametric uncertainties are unavoidable in realistic conditions. 

In the present paper, the model is assumed to be subject to varying parameters, varying time-delay, norm-bounded uncertainties and disturbances that impair the response of the closed-loop system to track a reference MAP profile. Hence, a robust time-delayed LPV gain-scheduled dynamic output-feedback controller is designed to guarantee robustness and tracking performance of the closed-loop system. The LMI framework is adopted to result in controller synthesis conditions in a convex and tractable setting using a Lyapunov-Krasovskii functional approach. Finally, the proposed robust LPV  control design method in conjunction with the MMSRCKF parameter estimation tool is validated via simulations.  Simulation results utilizing collected animal experiment data and a patient simulation model demonstrate the superiority and effectiveness of the control and estimation strategies to achieve MAP reference tracking, disturbance rejection, noise attenuation, and parametric uncertainty compensation.  

The notation to be used in the paper is standard and as follows. $\mathbb{R}$ denotes the set of real numbers, $\mathbb{R}_{+}$ is the set of non-negative real numbers, and $\mathbb{R}^n$ and $\mathbb{R}^{k \times m}$ are used to denote the set of real vectors of dimension $n$ and the set of real $k \times m$ matrices, respectively. $\mathbb{S}^{n}$ and $\mathbb{S}^{n}_{++}$ represent the set of real symmetric and real symmetric positive definite $n \times n$ matrices, respectively. $\mathbf{M} \succ \mathbf{0}$ shows the positive definiteness of the matrix $\mathbf{M}$. The inverse and transpose of a real matrix $\mathbf{M}$ are designated by $\mathbf{M}^{\text{T}}$ and $\mathbf{M}^{-1}$, respectively. $He [\mathbf{M}]$ is Hermitian operator defined as $He [\mathbf{M}] = \mathbf{M} + \mathbf{M}^{\text{T}}$. 
Also, In a symmetric matrix, the asterisk $\star$ in the $(i,\: j)$ element shows transpose of the $(j,\: i)$ element. $\mathscr{C} (J,\: K)$ stands for the set of continuous functions mapping a set $J$ to a set $K$. For a stochastic process, $\mathbf{x}_{k}$, $\mathscr{E}[\mathbf{x}_{k}]$ denotes its expected value and $\mathscr{N}\{\mathbf{x}_k;\widehat{\mathbf{x}}_{k|k},\mathbf{P}_{k|k} \}$ represents a normal Gaussian probability distribution with the mean of $\widehat{\mathbf{x}}_{k|k}$ and the covariance of $\mathbf{P}_{k|k}$.


The outline of the paper is as follows. Section \ref{sec:ProbForm} presents the mathematical description of the blood pressure dynamical model. The MMSRCKF parameter identification method is introduced in section \ref{sec:Estimation}, followed by the estimation results in section \ref{sec:EstimationResults}. In section \ref{sec:Control}, the LPV model of the MAP response is introduced and the robust time-delayed LPV gain-scheduling control design is described. Section \ref{sec:Results} outlines the simulation results and presents the evaluation of the performance of the proposed controller. Final remarks are provided in section \ref{sec:Conclusion}.

\section{MAP drug response model}\label{sec:ProbForm}

In this paper in line with the previous work in the literature (see \cite{cao2017simulator, sandu2016reinforcement,luspay2016adaptive, StasoujianRobust}) a first-order model with a time delay is considered to describe the patient's MAP response to the infusion of a vasoactive drug, such as phenylephrine (PHP), \textit{i.e.}

\beq
T(t) \cdot \dot{\Delta MAP}(t) + \Delta MAP(t) = K(t) \cdot u(t-\tau(t)),
\label{eq:MAP response TF}
\eeq
where  $\Delta MAP(t)$  stands for the MAP variations in $mmHg$ from its baseline value, \textit{i.e.} $\Delta MAP (t)= MAP (t) -MAP_b(t)$, $u(t)$ is the drug delivery rate in $ml/h$, $K(t)$ denotes the patient's sensitivity to the drug, $T(t)$ is the lag time representing the uptake, distribution and biotransformation of the drug \cite{isaka1993control}, and $\tau (t)$ is the time delay for the drug to reach the circulatory system from the infusion pump. This first-order model seems to properly capture a patient's physiological response to the PHP drug injection. Figure \ref{fig:MAPresponse} presents a typical MAP response due to a step PHP infusion versus a matched response of (\ref{eq:MAP response TF}). The figure also shows the interpretation of the model parameters $K(t)$, $T(t)$, $\tau(t)$, $MAP_b(t)$ which have been obtained to fit the MAP response using a least-squares optimization method. Data is collected from swine experiments performed at the Resuscitation Research Laboratory at the University of Texas Medical Branch (UTMB), Galveston, Texas \cite{luspay2016adaptive}. Although the proposed model structure (\ref{eq:MAP response TF}) is qualitatively able to represent the characteristics of the MAP response to the infusion of PHP, the model parameters vary considerably over time due to the  variability of patients' pharmacological response to the vasoactive drug infusion. That is, the model parameters and delay could vary significantly from patient-to-patient (inter-patient variability), as well as, for a given patient over time (intra-patient variability) \cite{isaka1993control, rao2003experimental}. 

\begin{figure}[!t] 
\hspace*{-.1in}
\centering \includegraphics[width=1.05\columnwidth, height=2.00in]{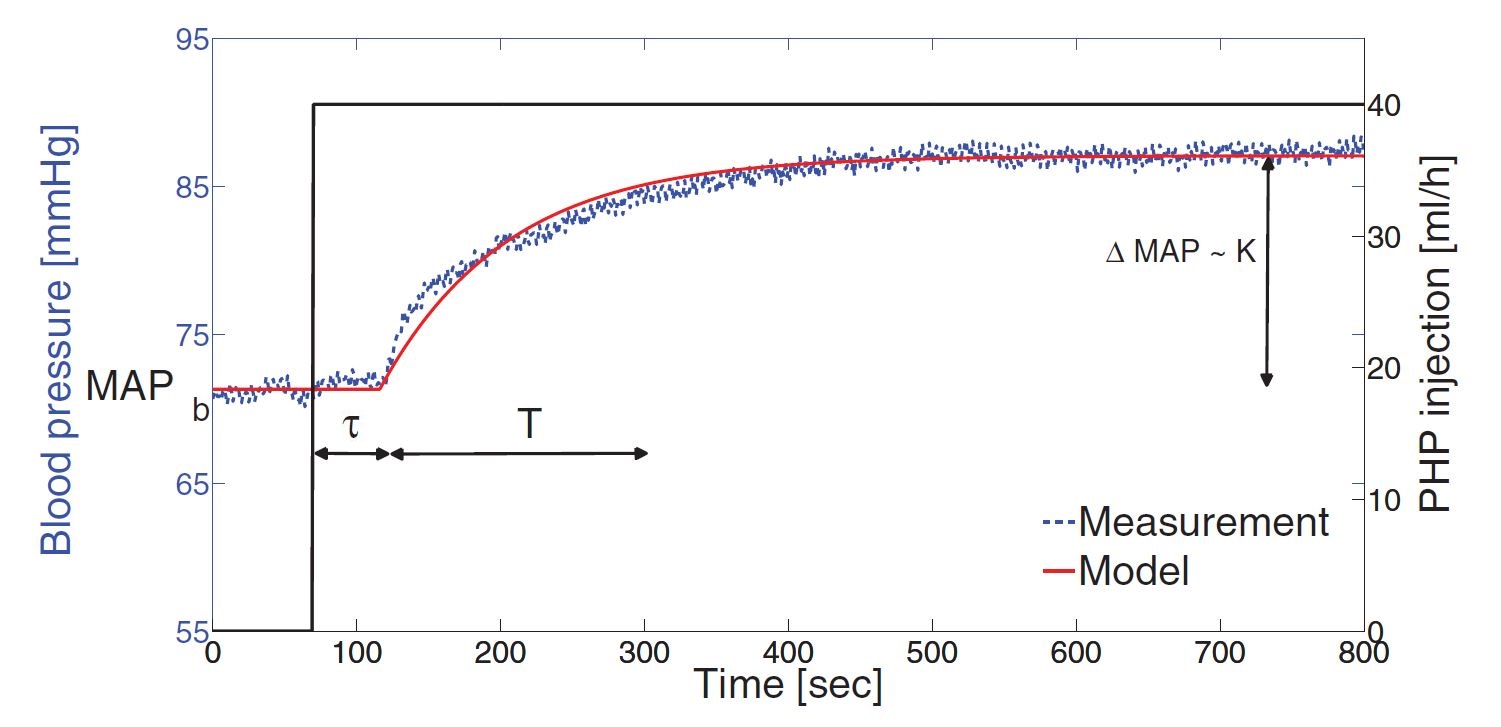} 
\caption{Typical MAP response due to step vasopressor drug infusion} 
\label{fig:MAPresponse}
\end{figure}

In the next section, a multiple-model square-root cubature Kalman filter (MMSRCKF) estimation algorithm is proposed and validated for the online estimation of the MAP response model parameters.

\section{Estimation preliminaries and methodology}\label{sec:Estimation}

		To implement the estimation framework, the continuous-time model (\ref{eq:MAP response TF}) is discretized at a sampling rate of $T_s$. Thus, the governing dynamics in discrete-time is given by
		\beq
			\left\{\begin{array}{l}
				x_{k+1} = \big(1-\dfrac{T_s}{T_k} \big) x_k + \dfrac{K_k T_s}{T_k} u_{(k-\frac{\tau_k}{T_s})},\\ [5pt]
				y_k = x_k + MAP_{b_k},
				\end{array}\right.
		\label{eq:DiscretizedEq}
		\eeq

		\noindent where $x_k = \Delta MAP_k = MAP_k - MAP_{b_k}$ at the $k$th time interval.  The state equation (\ref{eq:DiscretizedEq}) is augmented with the parameters to be estimated, namely $K_k, T_k,$ and $MAP_{b_k}$ to form an augmented state vector by assuming local random-walk dynamics. The state vector to be estimated is thus given by
		\begin{align}
			\mathbf{X}_k &= [\begin{array}{cccc}
								{X}_k^1 & {X}_k^2 & {X}_k^3 & {X}_k^4\end{array}]^{\text{T}}\nonumber\\
								& = [\begin{array}{cccc}
								{\Delta MAP}_k & K_k & T_k & MAP_{b_k}\end{array}]^{\text{T}}.
		\end{align}
		Since all model parameters are time-varying and assumed to be a \textit{priori} unknown, (\ref{eq:DiscretizedEq}) represents a nonlinear equation with regards to the state vector, $\mathbf{X}_k$, which can be expressed as the following nonlinear dynamics
		\beq
			\left\{\begin{array}{l}
			X_{k+1}^1 = \mathbf{f}_k(\mathbf{X}_k,u_k) + w_k,\\[3pt]
			y_k = h_k(\mathbf{X}_k) + v_k,
			\end{array}\right.
		\eeq
	with 
		\beq
			\left\{\begin{array}{l}
			f_k^1(\mathbf{X}_k,u_k) = \big(1-\dfrac{T_s}{X_k^3} \big) X_k^1 + \dfrac{T_s X_k^2}{X_k^3} u_{(k-\frac{\tau_k}{T_s})},\\[6pt]
			h_k(\mathbf{X}_k) = X_k^1 + X_k^4.
\end{array}			 \right.
		\eeq
	The process noise, $w_k$, and the measurement noise, $v_k$, are both assumed additive and statistically independent zero-mean Gaussian processes with covariances given by $\mathbf{Q}_{k}$ and ${R}_{k}$, respectively. As a consequence, linear regression methods like recursive least-squares and instrumental variables may fail in the efficient estimation of the parameters. Other local-approximation methods such as EKF require the model to be mildly nonlinear to be approximated via the first-order Taylor series. Moreover, partial derivatives of the nonlinear state-space model, \textit{i.e.} the Jacobians, must be computed which is not always viable. Therefore, these limitations motivated the use of a Bayesian-based filtering approach based on the cubature Kalman filter (CKF) through which the system's intrinsic nonlinear dynamics is employed directly \cite{haykin2009neural}. Although such an augmentation facilitates the estimation procedure, the time-varying input delay cannot be included in the augmented state vector or captured by a random walk process. Thus, it is computed through a multiple-model hypothesis testing process along with the CKF, which will be discussed later. 
	
\subsection{Square-root CKF }
	In the Bayesian-based CKF method, a probability approach is followed to the state estimation of dynamic systems \cite{haykin2009neural}. Due to the fact that accumulated numerical errors can lead to an indefinite error covariance matrix, square-root CKF (SRCKF) will be examined to overcome this problem. In this method, the covariance matrix is decomposed using a factorization method, such as the Cholesky factorization \cite{loehr2014advanced}. Then, the third-degree spherical-radial rule is used to approximate the multidimensional integrals involved in the Bayesian filtering \cite{jia2013high}. 
 Consider the following general nonlinear discrete-time stochastic system

	\beq
		\left\{\begin{array}{l}
		\mathbf{x}_{k+1} = \mathbf{f}(\mathbf{x}_k,\mathbf{u}_k) + \mathbf{w}_k,\\
		\mathbf{y}_k = \mathbf{h}(\mathbf{x}_k,\mathbf{u}_k) + \mathbf{v}_k,\; k=0,1,\ldots, k_f,
		\end{array}\right.
	\eeq
	
\noindent where $\mathbf{x}_k \in \mathbb{R}^n$ is the state vector or the unmeasurable states of the system, $\mathbf{u}_k \in \mathbb{R}^{n_u}$ is the input vector, and $\mathbf{y}_k \in \mathbb{R}^{n_y}$ is the measurement vector at the time $k$, and $k_f$ is the final time. The mappings $\mathbf{f}(\mathbf{x}_k,\mathbf{u}_k): (\mathbb{R}^n,\mathbb{R}^{n_u}) \mapsto \mathbb{R}^n$ and $\mathbf{h}(\mathbf{x}_k,\mathbf{u}_k): (\mathbb{R}^n,\mathbb{R}^{n_u}) \mapsto \mathbb{R}^{n_y}$ are known and the vectors $\mathbf{w}_k \in \mathbb{R}^n$ and $\mathbf{v}_k \in \mathbb{R}^{n_y}$ denote the process and measurement noise, respectively and are assumed mutually independent. The probability distribution functions (PDFs) of the noise, namely $p(\mathbf{w}_k)$ and $p(\mathbf{v}_k)$ are assumed to be known, as well as, the initial state PDF given by $p(\mathbf{x}_0)$.

CKF seeks to find the estimation of the state vector in the form of a conditional PDF, $p(\mathbf{x}_k|\mathbf{y}^k)$ where $\mathbf{y}^k=[\begin{array}{cccc} \mathbf{y}_0 & \mathbf{y}_1 & \ldots & \mathbf{y}_k\end{array}]$ denotes the vector of the measurements. However, in some cases, a Gaussian approximation of the conditional PDF allows to only compute the first two conditional moments, \textit{i.e.} the mean $\widehat{\mathbf{x}}_{k|k} = \mathscr{E}[\mathbf{x}_k|\mathbf{y}^k]$ and the error covariance matrix $\mathbf{P}_{k|k} = cov[\mathbf{x}_k|\mathbf{y}^k]$ which results in $p(\mathbf{x}_k|\mathbf{y}^k) \approx \mathscr{N}\{\mathbf{x}_k;\widehat{\mathbf{x}}_{k|k},\mathbf{P}_{k|k} \}$.

The third-degree spherical-radial rule is utilized in the CKF procedure to compute the moment integrals. Consequently, if the noise signal enters the system as Gaussian white noise, the prediction step (state prediction) and correction step (measurement update) are carried out via integrating a nonlinear function with regards to a normal distribution, that is

		{\small
		\begin{align}
			\widehat{\mathbf{x}}_{k+1|k}\! & =\! \mathscr{E}[\mathbf{x}_{k+1}|\mathbf{y}^k] \!\! =\!\! \int_{\mathbb{R}_n}\!\!\!\!\! \mathbf{f}(\mathbf{x}_k,\mathbf{u}_k) p(\mathbf{x}_{k}|\mathbf{y}^k) \text{d}\mathbf{x}_k\nonumber\\
			& \!\! \approx \int_{\mathbb{R}_n}\!\!\!\!\! \mathbf{f}(\mathbf{x}_k,\mathbf{u}_k) \mathscr{N}\{\mathbf{x}_k;\widehat{\mathbf{x}}_{k|k},\mathbf{P}_{k|k} \}\text{d}\mathbf{x}_k,\\
			\widehat{\mathbf{y}}_{k+1|k} \! & = \! \mathscr{E}[\mathbf{y}_{k+1}|\mathbf{x}_{k+1}]\!\! =\!\! \int_{\mathbb{R}_n}\!\!\!\!\!\! \mathbf{h}(\mathbf{x}_{k+1},\mathbf{u}_{k+1}) p(\mathbf{y}_{k+1}|\mathbf{x}_{k+1}) \text{d}\mathbf{x}_{k+1}\nonumber\\
			& \!\! \approx \int_{\mathbb{R}_n}\!\!\!\!\! \mathbf{h}(\mathbf{x}_{k+1},\mathbf{u}_{k+1}) \mathscr{N}\{\mathbf{x}_{k+1};\widehat{\mathbf{x}}_{k+1|k},\mathbf{P}_{k+1|k} \}\text{d}\mathbf{x}_{k+1}.
		\end{align}
	}Next, for an arbitrary function $g(\mathbf{x})$ with $\boldsymbol{\Sigma}$ as the covariance of $\mathbf{x}$, the integral
\begin{small}
		\beq
			I(g) = \sqrt{2\pi} \vert\boldsymbol{\Sigma} \vert^{-\frac{1}{2}}\int_{\mathbb{R}^n}\!\!\! g(\mathbf{x})exp\left[-\dfrac{1}{2}(\mathbf{x}-\boldsymbol{\mu})^\text{T} \boldsymbol{\Sigma}^{-1}(\mathbf{x}-\boldsymbol{\mu}) \right]\text{d}\mathbf{x},
		\eeq
\end{small}\noindent can be expressed in the spherical coordinate system as 
		\beq
			I(g) = (2\pi)^{-\frac{n}{2}} \int_{r=0}^\infty \int_{\mathbb{U}_n} g(\mathbf{C}r\mathbf{z}+\boldsymbol{\mu})\text{d}\mathbf{z} r^{n-1}e^{-\frac{r^2}{2}} \text{d}r,
		\eeq
	where $\mathbf{x} = \mathbf{C}r\mathbf{z}+\boldsymbol{\mu}$ with $\Vert\mathbf{z} \Vert =1$, $\boldsymbol{\mu}$ is the mean and $\mathbf{C}$ is the Cholesky factor of the covariance, $\boldsymbol{\Sigma}$, and $\mathbb{U}_n$ is the unit sphere. Then, the symmetric spherical cubature rule is used to further approximate the integral through the following relation
		\beq
			I(g) = \dfrac{1}{2n} \sum\limits_{i=0}^{2n} g(\sqrt{n}(\mathbf{C}{\xi}_i+\boldsymbol{\mu})),
		\eeq
	where ${\xi}_i$ denotes the $i$th cubature point at the intersection of the unit sphere and its axes. The main advantage of this method is that the cubature points are obtained off-line using a third-degree cubature rule \cite{liu2014adaptive}. Hence, one can use the following steps to compute state estimation using SRCKF.

\subsubsection*{SRCKF algorithm}
\begin{enumerate}
	\item {\bf Initialization}: The state initial condition is given by $\mathbf{x}_{0|0}\equiv \mathbf{x}_{0}$ with $\widehat{\mathbf{x}}_0=\mathscr{E}[\mathbf{x}_0]$ where the initial covariance matrix is $\mathbf{P}_{0|0}$ which is decomposed as $\mathbf{P}_{0|0} = \mathbf{S}_{0|0}\mathbf{S}_{0|0}^\text{T}$ through Cholesky factorization, \textit{i.e.}
		\benn
			\mathbf{S}_{0|0} = chol\{[\mathbf{x}_0 - \widehat{\mathbf{x}}_0][\mathbf{x}_0 - \widehat{\mathbf{x}}_0]^\text{T}\}.
		\eenn
		Then, the cubature points, $\xi_i$, and the weights, $w_i = w = \dfrac{1}{2n}$, are set for $i=1,2,\ldots,2n$.

\item \textbf{Time update (Prediction)} $(k=1,2,\ldots, k_f)$:
		\begin{enumerate}
			\item Evaluation of the cubature points
				\beq
					\mathbf{X}_{i,k-1|k-1} = \mathbf{S}_{k-1|k-1} \xi_i + \widehat{\mathbf{x}}_{k-1|k-1}.
				\eeq
			\item Evaluation of the propagated cubature points via the system dynamics
				\beq
					\mathbf{X}_{i,k|k-1}^* = \mathbf{f}_k(\mathbf{X}_{i,k-1|k-1},\mathbf{u}_{k-1}).
				\eeq
			\item Evaluation of the predicted states based on the weights and propagated points
				\beq
					\widehat{\mathbf{x}}_{k|k-1} = \sum\limits_{i=1}^{2n} w_i \mathbf{X}_{i,k|k-1}^*.
				\eeq
			\item Evaluation of the square-root of the covariance of the predicted state error covariance
				\beq
					\mathbf{S}_{k|k-1} = triangle\big\{[\boldsymbol{\chi}_{k|k-1}^*, \mathbf{S}_{\mathbf{Q}_{k-1}} ]\big\},
					\label{eq:S}
				\eeq
			where $\mathbf{B}=triangle\{\mathbf{A}\}$ stands for a general triangularization algorithm, \textit{e.g.} QR decomposition, where $\mathbf{B}$ is a lower triangular matrix. If $\mathbf{C}$ is an upper triangular matrix obtained through the QR decomposition of $\mathbf{A}^\text{T}$, then the lower triangular matrix is given by $\mathbf{B} = \mathbf{C}^{\text{T}}$. In (\ref{eq:S}), $\boldsymbol{\chi}_{k|k-1}^*$ is a centered, weighted matrix given by
				\begin{align}
					& \boldsymbol{\chi}_{k|k-1}^* = \dfrac{1}{\sqrt{2n}}[\mathbf{X}_{1,k|k-1}^* - \widehat{\mathbf{x}}_{k|k-1}\nonumber\\
										&\quad \begin{array}{ccc}
															 \mathbf{X}_{2,k|k-1}^* - \widehat{\mathbf{x}}_{k|k-1} &\cdots & \mathbf{X}_{2n,k|k-1}^* - \widehat{\mathbf{x}}_{k|k-1}					
\end{array}].
				\end{align}
				$\mathbf{S}_{\mathbf{Q}_{k-1}}$ is the square-root of the the process noise such that $\mathbf{Q}_{k-1}=\mathbf{S}_{\mathbf{Q}_{k-1}} \mathbf{S}_{\mathbf{Q}_{k-1}}^\text{T}$.
		\end{enumerate}
		
	\vspace{1mm}	
	\item \textbf{Measurement update (Correction)} $(k=1,2,\ldots, k_f)$:
		\begin{enumerate}
			\item Evaluation of the cubature points
				\beq
					\mathbf{X}_{i,k|k-1} = \mathbf{S}_{k|k-1} \xi_ i + \widehat{\mathbf{x}}_{k|k-1}.
				\eeq
			\item Evaluation of the propagated cubature point via the output dynamics
				\beq
					\mathbf{Y}_{i,k|k-1} = \mathbf{h}(\mathbf{X}_{i,k|k-1},\mathbf{u}_k).
				\eeq
			\item Estimation of the predicted measurement
				\beq
					\widehat{\mathbf{y}}_{k|k-1} = \sum\limits_{i=1}^{2n} w_i \mathbf{Y}_{i,k|k-1}.
				\eeq
			\item Evaluation of the square-root of the innovation covariance matrix
				\beq
					\mathbf{S}_{yy,k|k-1} = triangle\big\{[\mathbf{Y}_{k|k-1}, \mathbf{S}_{\mathbf{R}_{k}}]\big\},					
				\eeq
			where $\mathbf{Y}_{k|k-1}$ is a centered, weighted matrix given by
				\begin{align}
					& \mathbf{Y}_{k|k-1} = \dfrac{1}{\sqrt{2n}}[\mathbf{Y}_{1,k|k-1} - \widehat{\mathbf{y}}_{k|k-1}\nonumber\\
															& \quad	\begin{array}{ccc}
																\mathbf{Y}_{2,k|k-1} - \widehat{\mathbf{y}}_{k|k-1} &\cdots & \mathbf{Y}_{2n,k|k-1} - \widehat{\mathbf{y}}_{k|k-1}					
\end{array}].
				\end{align}
				$\mathbf{S}_{\mathbf{R}_{k}}$ is also the square-root of the the measurement noise such that $\mathbf{R}_{k}=\mathbf{S}_{\mathbf{R}_{k}} \mathbf{S}_{\mathbf{R}_{k}}^\text{T}$.	
				\vspace{1mm}					
			\item Evaluation of the cross-covariance matrix
				\beq
					\mathbf{P}_{xy,k|k-1} = \boldsymbol{\chi}_{k|k-1}\mathbf{Y}_{k|k-1}^\text{T},
				\eeq
			with the centered, weighted matrix $\boldsymbol{\chi}_{k|k-1}$ obtained by
				\begin{align}
					& \boldsymbol{\chi}_{k|k-1} = \dfrac{1}{\sqrt{2n}}[\mathbf{X}_{1,k|k-1} - \widehat{\mathbf{x}}_{k|k-1} \nonumber\\
					&\quad \begin{array}{ccc}
																\mathbf{X}_{2,k|k-1} - \widehat{\mathbf{x}}_{k|k-1} &\cdots & \mathbf{X}_{2n,k|k-1} - \widehat{\mathbf{x}}_{k|k-1}					
\end{array}].
				\end{align}
			\item Evaluation of the SRCKF filter gain
				\beq
					\mathbf{W}_k = \mathbf{P}_{xy,k|k-1} \mathbf{S}_{yy,k|k-1}^{-\text{T}} \mathbf{S}_{yy,k|k-1}^{-1}. 
				\eeq
			\item Evaluation of the corrected state update based on the measurement
				\beq
					\widehat{\mathbf{x}}_{k|k} = \widehat{\mathbf{x}}_{k|k-1} + \mathbf{W}_k (\mathbf{y}_k - \widehat{\mathbf{y}}_{k|k-1}).
				\eeq
			\item Evaluation of the square-root of the corrected error covariance matrix
				\beq
					\mathbf{S}_{k|k} = triangle\big\{[\boldsymbol{\chi}_{k|k-1}-\mathbf{W}_k\mathbf{Y}_{k|k-1}, \mathbf{W}_k\mathbf{S}_{\mathbf{R}_k}]\big\}.
				\eeq
		\end{enumerate}
\end{enumerate}

\noindent The state estimation process continues iteratively from the second step of the algorithm, \textit{i.e.} the time update (prediction) by setting $k = k + 1$. The flowchart depicting the SRCKF algorithm is shown in Fig. \ref{fig:SRCKF flowchart}.

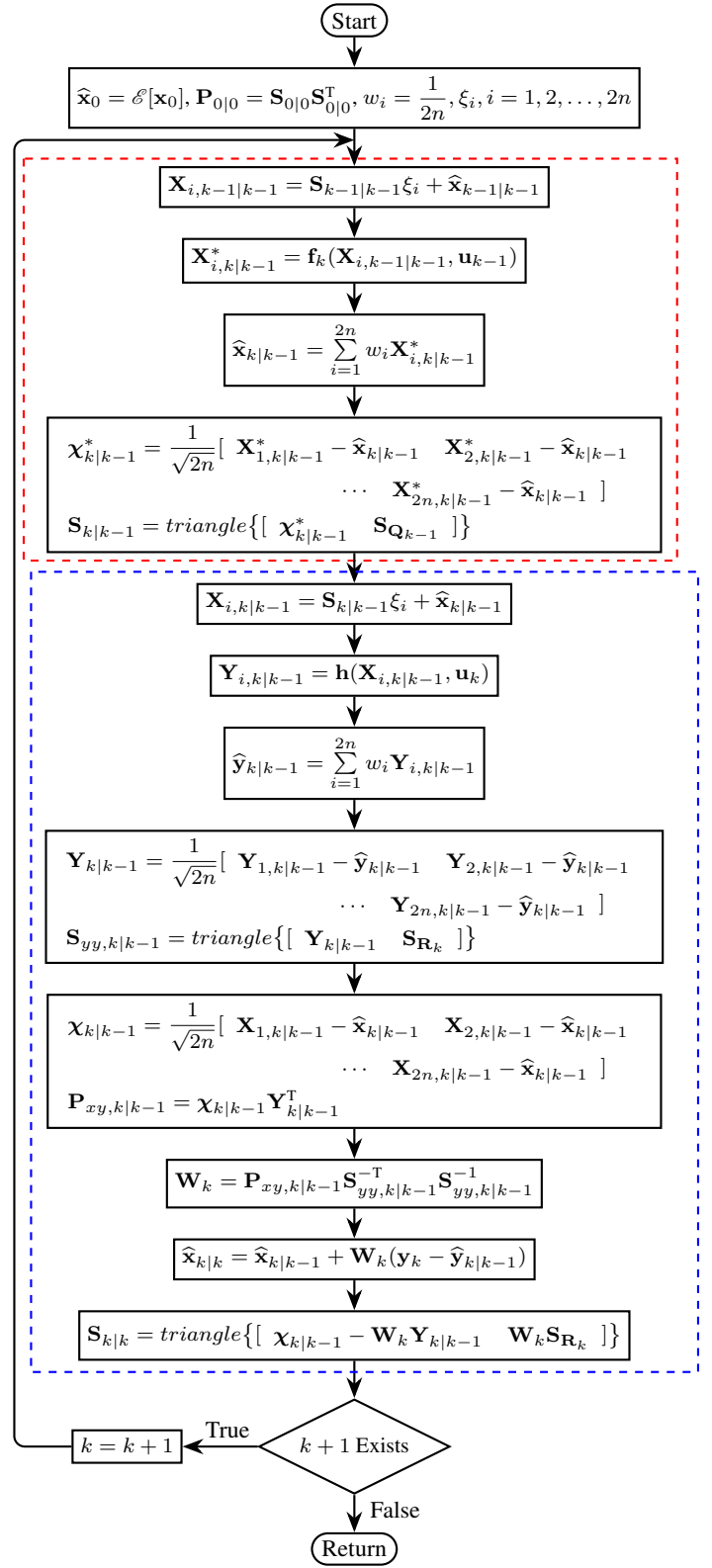
\begin{figure}[]

\begin{tikzpicture}

\node (start) [startstop] {Start};


\node (box1) [process,below=0.4 of start] {\footnotesize $\widehat{\mathbf{x}}_0=\mathscr{E}[\mathbf{x}_0]$, $\mathbf{P}_{0|0} = \mathbf{S}_{0|0}\mathbf{S}_{0|0}^\text{T}$, $w_i=\dfrac{1}{2n}, \xi_i, i=1,2,\ldots,2n$};

\node (box2) [process,below=0.5 of box1] {\footnotesize $\mathbf{X}_{i,k-1|k-1} = \mathbf{S}_{k-1|k-1} \xi_i + \widehat{\mathbf{x}}_{k-1|k-1}$};

\node (box3) [process,below=0.4 of box2] {\footnotesize $\mathbf{X}_{i,k|k-1}^* = \mathbf{f}_k(\mathbf{X}_{i,k-1|k-1},\mathbf{u}_{k-1})$};

\node (box4) [process,below=0.4 of box3] {\footnotesize $\widehat{\mathbf{x}}_{k|k-1} = \sum\limits_{i=1}^{2n} w_i \mathbf{X}_{i,k|k-1}^*$};

\node (box5) [process,below=0.4 of box4] {\footnotesize $\begin{array}{l} \boldsymbol{\chi}_{k|k-1}^* = 

	\dfrac{1}{\sqrt{2n}}[\begin{array}{cc} \mathbf{X}_{1,k|k-1}^* - \widehat{\mathbf{x}}_{k|k-1} & \mathbf{X}_{2,k|k-1}^* - \widehat{\mathbf{x}}_{k|k-1} \end{array}\\

	\begin{array}{cc}

	\qquad\qquad\qquad\qquad\qquad\qquad \cdots & \mathbf{X}_{2n,k|k-1}^* - \widehat{\mathbf{x}}_{k|k-1} \end{array}]\\

	\mathbf{S}_{k|k-1} = triangle\big\{[\begin{array}{cc} \boldsymbol{\chi}_{k|k-1}^* & \mathbf{S}_{\mathbf{Q}_{k-1}}\end{array} ]\big\}\end{array}$};

\draw[red,thick,dashed] ($(box2.north west)+(-1.8,0.1)$)  rectangle ($(box5.south east)+(0.2,-0.1)$);

\node (box6) [process,below=0.4 of box5] {\footnotesize $\mathbf{X}_{i,k|k-1} = \mathbf{S}_{k|k-1} \xi_ i + \widehat{\mathbf{x}}_{k|k-1}$};

\node (box7) [process,below=0.4 of box6] {\footnotesize $\mathbf{Y}_{i,k|k-1} = \mathbf{h}(\mathbf{X}_{i,k|k-1},\mathbf{u}_k)$};

\node (box8) [process,below=0.4 of box7] {\footnotesize $\widehat{\mathbf{y}}_{k|k-1} = \sum\limits_{i=1}^{2n} w_i \mathbf{Y}_{i,k|k-1}$};

\node (box9) [process,below=0.4 of box8] {\footnotesize $ \begin{array}{l}  \mathbf{Y}_{k|k-1} = 

	\dfrac{1}{\sqrt{2n}}[\begin{array}{cc} \mathbf{Y}_{1,k|k-1} - \widehat{\mathbf{y}}_{k|k-1} & \mathbf{Y}_{2,k|k-1} - \widehat{\mathbf{y}}_{k|k-1} \end{array}\\

	\begin{array}{cc}

	\qquad\qquad\qquad\qquad\qquad\qquad \cdots & \mathbf{Y}_{2n,k|k-1} - \widehat{\mathbf{y}}_{k|k-1} \end{array}] \\

	\mathbf{S}_{yy,k|k-1} = triangle\big\{[\begin{array}{cc}\mathbf{Y}_{k|k-1} & \mathbf{S}_{\mathbf{R}_{k}}\end{array} ]\big\}\end{array}$};

\node (box10) [process,below=0.4 of box9] {\footnotesize $\begin{array}{l}  \boldsymbol{\chi}_{k|k-1} = 

	\dfrac{1}{\sqrt{2n}}[\begin{array}{cc} \mathbf{X}_{1,k|k-1} - \widehat{\mathbf{x}}_{k|k-1} &	\mathbf{X}_{2,k|k-1} - \widehat{\mathbf{x}}_{k|k-1}  \end{array}\\ \begin{array}{cc}

	\qquad\qquad\qquad\qquad\qquad\qquad \cdots & \mathbf{X}_{2n,k|k-1} - \widehat{\mathbf{x}}_{k|k-1}\end{array}]\\

	\mathbf{P}_{xy,k|k-1} = \boldsymbol{\chi}_{k|k-1}\mathbf{Y}_{k|k-1}^\text{T}\end{array}$};	

\node (box11) [process,below=0.4 of box10] {\footnotesize $\mathbf{W}_k = \mathbf{P}_{xy,k|k-1} \mathbf{S}_{yy,k|k-1}^{-\text{T}} \mathbf{S}_{yy,k|k-1}^{-1}$};	

\node (box12) [process,below=0.4 of box11] {\footnotesize $\widehat{\mathbf{x}}_{k|k} = \widehat{\mathbf{x}}_{k|k-1} + \mathbf{W}_k (\mathbf{y}_k - \widehat{\mathbf{y}}_{k|k-1})$};	

\node (box13) [process,below=0.4 of box12] {\footnotesize $\mathbf{S}_{k|k} = triangle\big\{[\begin{array}{cc}\boldsymbol{\chi}_{k|k-1}-\mathbf{W}_k\mathbf{Y}_{k|k-1} & \mathbf{W}_k\mathbf{S}_{\mathbf{R}_k}\end{array} ]\big\}$};

\node (branch1) [decision,aspect=2,below=0.5 of box13] {\footnotesize $k+1$ Exists};

\node (box14) [process,left=1 of branch1] {\footnotesize $k=k+1$};

\node (return) [startstop,below=.5 of branch1] {Return};






\draw [arrow] (start) -- (box1);

\draw [arrow] (box1) -- coordinate[pos=.3](m1)(box2);

\draw [arrow] (box2) -- (box3);

\draw [arrow] (box3) -- (box4);

\draw [arrow] (box4) -- (box5);

\draw [arrow] (box5) -- (box6);

\draw [arrow] (box6) -- (box7);

\draw [arrow] (box7) -- (box8);

\draw [arrow] (box8) -- (box9);

\draw [arrow] (box9) -- (box10);

\draw [arrow] (box10) -- (box11);

\draw [arrow] (box11) -- (box12);

\draw [arrow] (box12) -- (box13);

\draw [arrow] (box13) -- (branch1);

\draw [arrow] (branch1) -- coordinate[pos=0.4](m3)(return);

\node [black,right=0.1 of m3] {False};

\draw [arrow] (branch1) -- node[pos=.4,above]{True}(box14);

\draw [arrow] (box14) -- ++(-1.5,0)coordinate[pos=0.4](m5) |- (m1);

\draw[blue,thick,dashed] ($(box6.north west)+(-2.2,0.15)$)  rectangle ($(box13.south east)+(0.9,-0.15)$);
\end{tikzpicture}
    \caption{SRCKF algorithm flowchart}
    \label{fig:SRCKF flowchart}
\end{figure}

\subsection{Multiple-model SRCKF for time-delay estimation}
Time-delay estimation introduces a challenge in the estimation framework since the variable delay cannot be transformed into an equivalent random walk process. Rational approximations of the delay may be used, such as Pad\'{e} approximation; however, the introduced error may be significant, especially for large and time-varying delays. Thus, in order to obtain a more accurate delay estimation, the previously introduced SRCKF algorithm is equipped with a multiple-model framework with a hypothesis testing \cite{hanlon2000multiple}.

The underlying idea of the multiple-model SRCKF (MMSRCKF) method is to use a bank of $N$ identical SRCKFs in a parallel setting, as shown in Fig. \ref{fig:Bank}. Every filter uses the same measurement and input data, but a different delay is assigned to each element. The $i$th element in the bank provides us with a state estimation $\mathbf{X}_k^i$ together with the residuals $\mathbf{r}_k^i = \mathbf{y}_k - \widehat{\mathbf{y}}_k^i$. Having this information, a hypothesis testing method can then be used to obtain information on the value of the delay. Specifically, if the delay matches the one assigned to the $i$th SRCKF element, then the corresponding residual is essentially a zero-mean white noise process, \textit{i.e.} $\mathscr{E}[\mathbf{r}_k^i] = 0$, and its covariance given by
	\begin{align}
		\mathscr{E}[\mathbf{r}_k^i (\mathbf{r}_k^i)^\text{T}] & = \mathbf{H}\mathbf{P}_k^i \mathbf{H} + \mathbf{R} \triangleq \mathbf{R}_k^i.
	\end{align}
\noindent where $\mathbf{H}=[1 \; 0 \; 0 \; 1]$, $\mathbf{P}_k^i$ denotes the estimation covariance at the $k$th step, and $\mathbf{R}$ denotes the measurement noise covariance. The conditional probability density function of the $i$th SRCKF element measurement can be computed through
	\beq
		f(\widehat{y}_k^i|y_k) = \dfrac{1}{(2\pi)^{\frac{m}{2}}\vert \mathbf{R}_k^i\vert^{\frac{1}{2}}}exp\Big\{-\dfrac{1}{2}(\mathbf{r}_k^i)^\text{T} (\mathbf{R}_k^i)^{-1} \mathbf{r}_k^i \Big\},
	\eeq
where $m$ is the dimension of available measurements at each time step. Then, the conditional probability of each hypothesis is
	\beq
		p_k^i = \dfrac{f(\widehat{y}_k^i|y_k)p_{k-1}^i}{\sum\limits_{j=1}^N f(\widehat{y}_k^j|y_k)p_{k-1}^j},
	\eeq
\noindent where $p_k^i$ can be interpreted as the normalized conditional probability of a case when the delay equals the assigned value to the $i$th filter, \textit{i.e.} $\sum\limits_{j=1}^N p_{k}^j=1$. Now, it is possible to estimate the delay according to the element which has the highest probability. However, to obtain a more accurate delay estimation and avoid large fluctuations, instead of choosing the most likely delay estimation, we use the probabilities as weights to blend the hypotheses resulting from a number of filters. In other words, the time delay can be estimated as
	\beq
		\hat{\tau}_k^{MM} = \sum\limits_{j=1}^N p_{k}^j \tau_k^j,
	\eeq	
\noindent where $\tau_k^j$ is the delay estimation of the $i$th filter. In the following section, the bank of $N$ parallel SRCKF estimators of the MMSRCKF (see Fig. \ref{fig:Bank}) will be implemented for the model parameter and the time delay estimation of the MAP response dynamics.

\begin{figure}[]
	\centering
		\begin{tikzpicture}[auto, font = {\sf \scriptsize}, cross/.style={path picture={\draw[black] (path picture bounding box.south east) -- (path picture bounding box.north west) (path picture bounding box.south west)-- (path picture bounding box.north east);}}, node distance=2cm,>=latex]


		\node at (-5,0)[input, name=input1] {};

		\node at (-4.5,0)[input, name=inputm] {};		

		\node at (-3,1)[tip, name=input2] {\Large $\mathbf{\vdots}$};

		\node at (-3,-1)[tip, name=input3] {\Large $\mathbf{\vdots}$};		

		\node at (-3,0)[block,text width=1.7cm,text height=.0em,text depth=0em] (CKFi) {\centering SRCKF with $\tau_i$};

 	    \node at (-3,2)[block,text width=1.73cm,text height=.0em,text depth=.0em] (CKF1) {\centering SRCKF with $\tau_1$};

 	    \node at (-3,-2)[block,text width=1.8cm,text height=.0em,text depth=.0em] (CKFN) {\centering SRCKF with $\tau_N$};

		\node at (0,0)[block,text width=2.1cm,text height=.7em,text depth=.2em] (HT) {\centering Hypothesis Testing};
	    \node at (2,0)[output, node distance=2.1cm] (output1) {};
	    \draw [-] (input1) -- node[name=in1tosum1, pos=0.5, above] {$\mathbf{u}_k$} node[name=ms2, pos=.5, below] {$\mathbf{y}_k$} (inputm);

		\draw [->] (inputm) -- node[name=sum1tosum2, pos=0.4] {} (CKF1.west);	 	    

	    \draw [->] (inputm) -- node[name=sum2toCont, pos=0.4] {} (CKFi);

	    \draw [->] (inputm) -- node[name=sum2toCont, pos=0.4] {} (CKFN.west);

	    \draw [->] (CKF1.east) -- node[name=sum2toCont, pos=0.5, right] {$\widehat{\mathbf{X}}_k^1$} (HT.170);

	    \draw [->] (CKFi) -- node[name=sum2toCont, pos=0.5, above] {$\widehat{\mathbf{X}}_k^i$} (HT);

	    \draw [->] (CKFN.east) -- node[name=sum2toCont, pos=0.5, right] {$\widehat{\mathbf{X}}_k^N$} (HT.190);

    		\draw [->] (HT) -- node[name=sum2toCont, pos=0.6, above] {$\widehat{\tau}_k^{MM}$} (output1);

    		\draw[red,thick,dashed] ($(CKF1.north west)+(-0.15,.15)$)  rectangle ($(CKFN.south east)+(0.15,-0.15)$);

	\end{tikzpicture}

    \caption{Bank of $N$ parallel SRCKFs for delay estimation}

    \label{fig:Bank}

\end{figure}
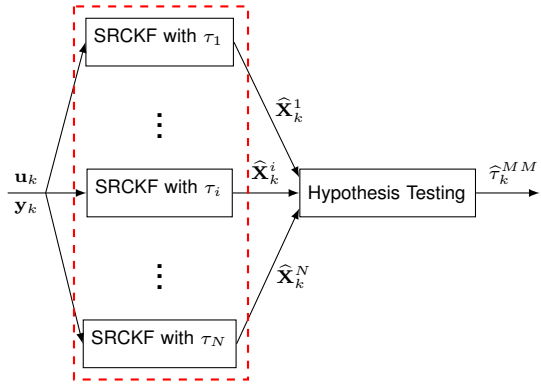

\section{MAP respose estimation}\label{sec:EstimationResults}

Experimental data from anesthetized swine acquired at the Resuscitation Research Laboratory, Department of Anesthesiology, UTMB in Galveston, Texas are utilized for the validation of MAP dynamic model parameter estimation using the proposed MMSRCKF method. An intramuscular injection of ketamine was used to sedate the swine which were maintained under anesthetic conditions by the continuous infusion of propofol. In order to monitor the blood pressure, a Philips MP2 transport device with a sampling frequency of $20$ Hz was used, while the PHP drug was infused through a bodyguard infusion pump. The $6$-hour experiment was performed on a swine of $55$ kg. Fig. \ref{fig:Bloodpressure} depicts the piece-wise constant PHP drug injection profile versus the corresponding absolute blood pressure response over time. To implement the estimation process, the experimental data has been re-sampled with a sampling frequency of $0.2$ Hz.  

\begin{figure}[!t] 
\hspace*{-.17in}
\centering \includegraphics[width=\columnwidth, height=2.4in]{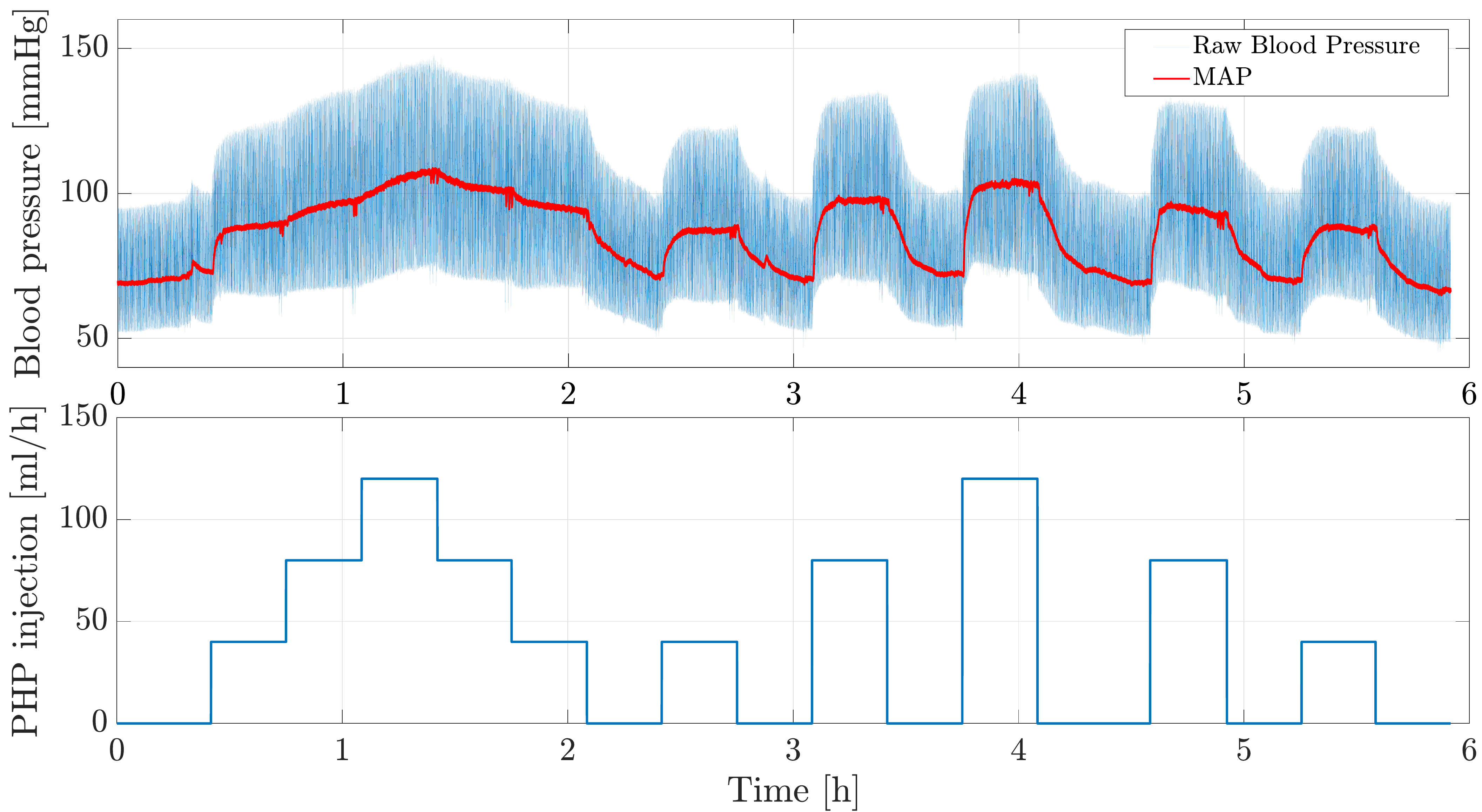} 
\caption{Experimental instantaneous blood pressure and MAP response to a piece-wise constant PHP drug infusion} 
\label{fig:Bloodpressure}
\end{figure}

\begin{figure}[!t] 
\hspace*{-.17in}
\centering \includegraphics[width=\columnwidth, height=2.1in]{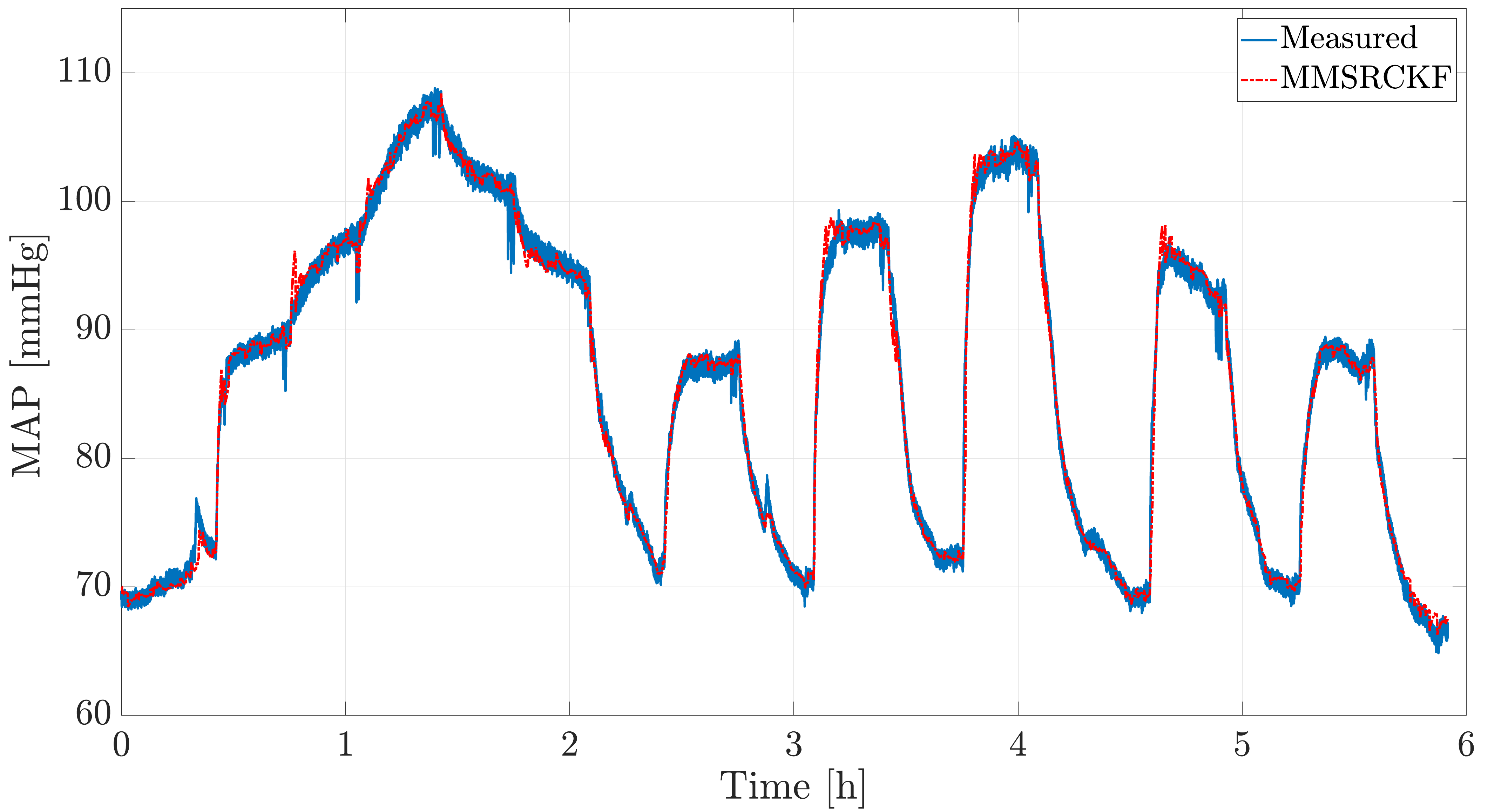} 
\caption{MAP estimation results} 
\label{fig:MAPestimation}
\end{figure}

To effectively capture the delay using the proposed MMSRCK algorithm and to address the trade-off between the delay estimation accuracy and the speed of convergence, a bank of $11$ SRCKFs with a delay interval of $\tau (t) \in [0\;100] s$ is considered. As a result, the time gridding for the evenly distributed filters is equal to $10 s$. The results of the implemented estimation approach on experimental data, as well as, the clinically acquired MAP measurements are shown in Fig. \ref{fig:MAPestimation}. As per the figure, the estimation method is capable of precisely capturing the MAP response of the patient to the injection of the vasoactive drug. Moreover, the estimation of the model parameters, namely the  sensitivity $K (t)$, time constant $T (t)$, MAP baseline value $MAP_b (t)$, and time delay $\tau (t)$, are shown in Figs. \ref{fig:Par_Sensitivity}, \ref{fig:Par_Lag}, \ref{fig:Par_Baseline}, and \ref{fig:Par_Delay}, respectively. The estimated parameter values follow the expected trends as discussed in detail in \cite{StasoujianRobust}. Moreover, the delay estimation shown in Fig. \ref{fig:Par_Delay} demonstrates a sharp initialization peak right after the initial injection of the drug and follows a slowly decaying trend during the rest of the experiment as expected \cite{Craig2004}.

\begin{figure}[!t] 
\hspace*{-.17in}
\centering \includegraphics[width=\columnwidth, height=2.15in]{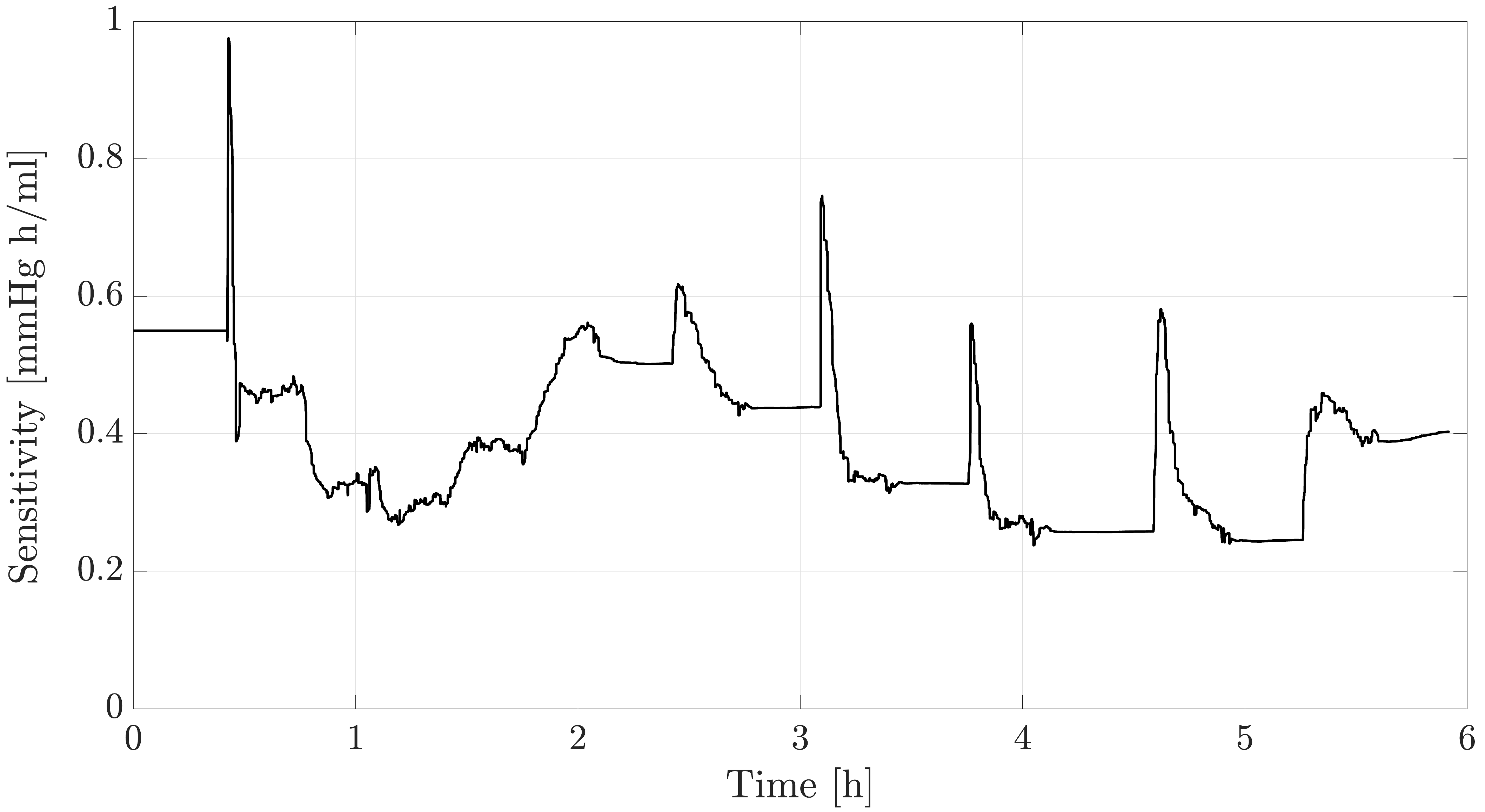} 
\caption{Sensitivity parameter estimation} 
\label{fig:Par_Sensitivity}
\end{figure}

\begin{figure}[!t] 
\hspace*{-.17in}
\centering \includegraphics[width=\columnwidth, height=2.15in]{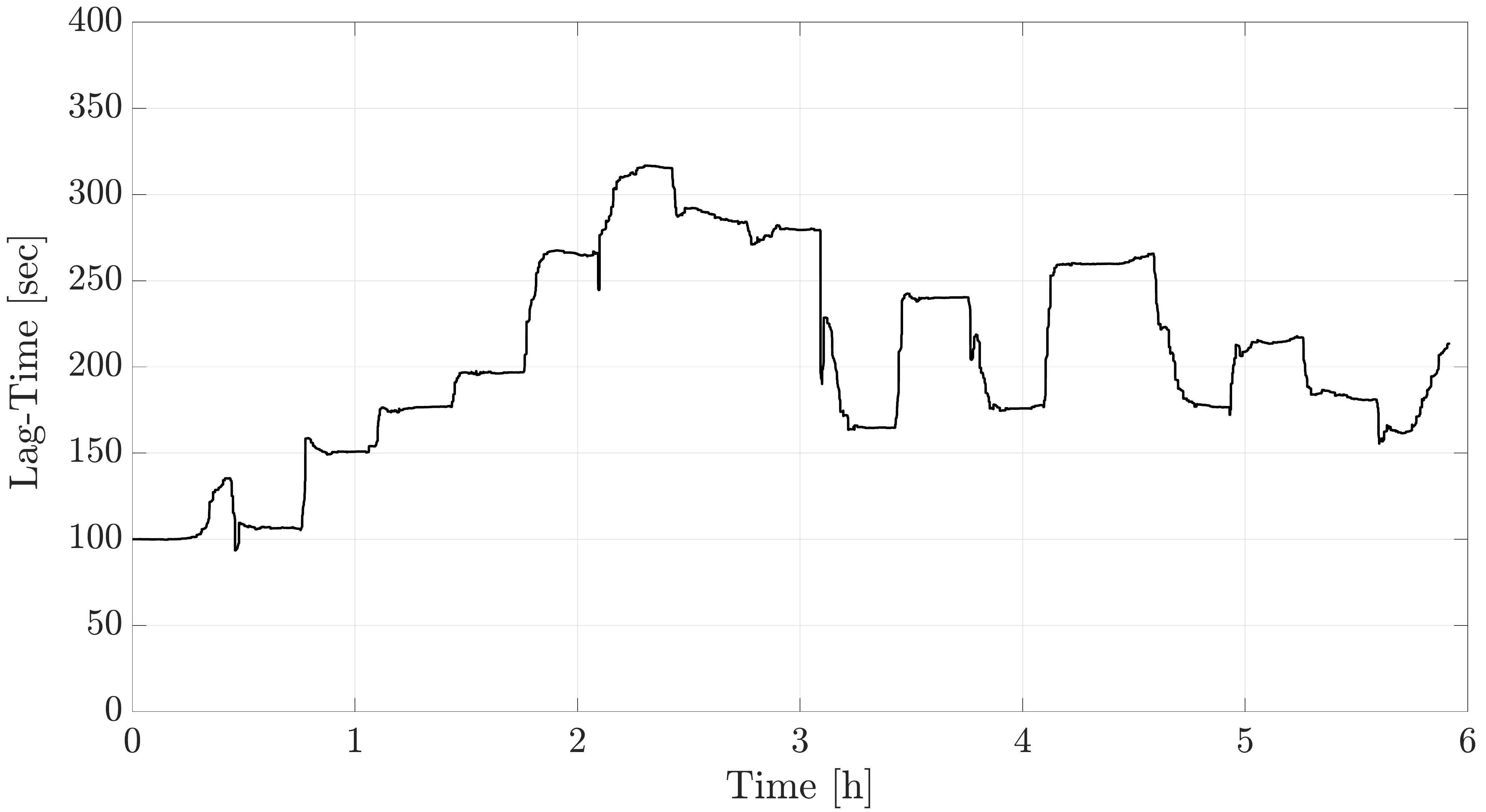} 
\caption{Lag-time parameter estimation} 
\label{fig:Par_Lag}
\end{figure}

\begin{figure}[!t] 
\hspace*{-.165in}
\centering \includegraphics[width=0.985\columnwidth, height=2.15in]{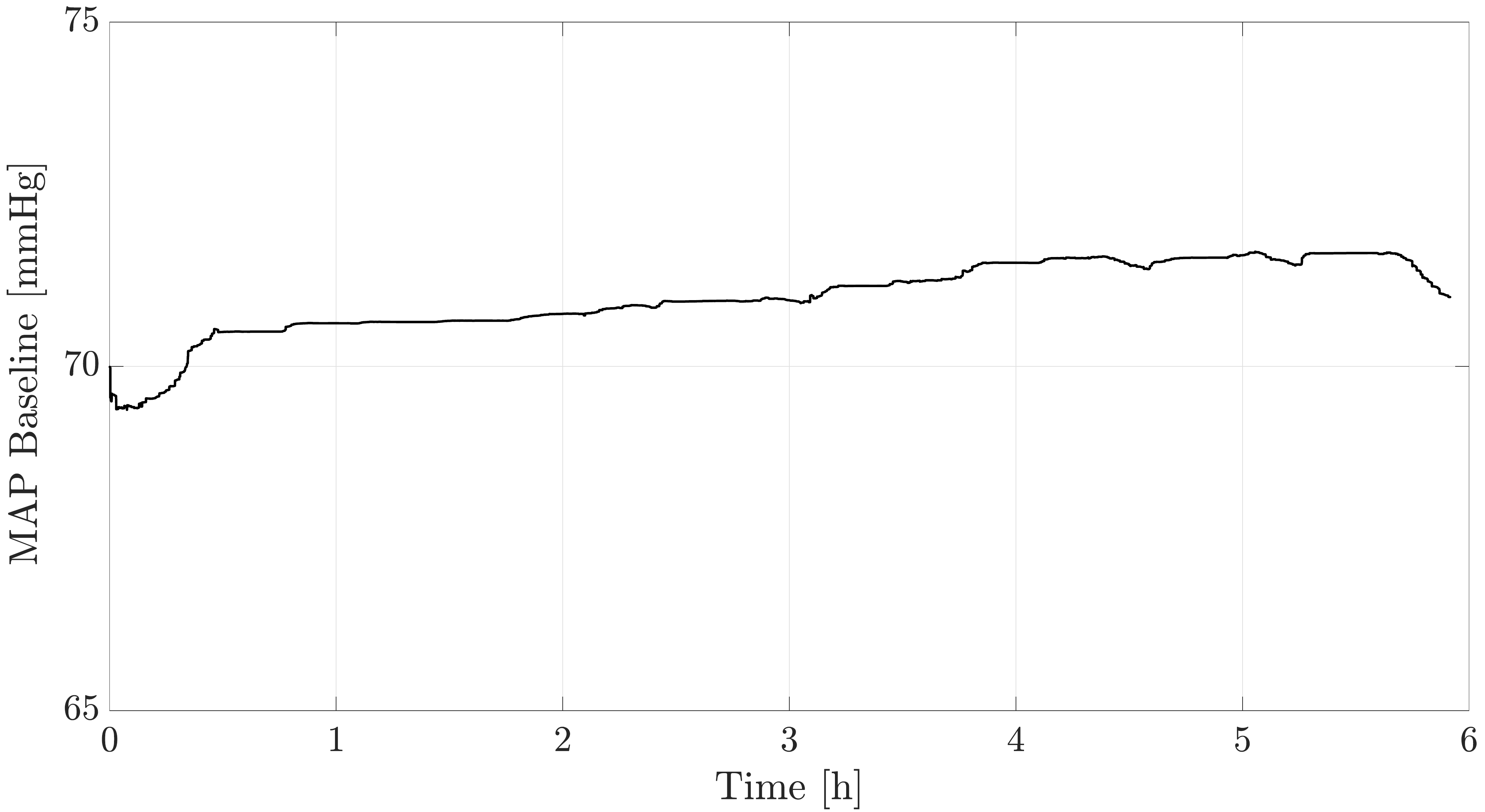} 
\caption{Baseline MAP parameter estimation} 
\label{fig:Par_Baseline}
\end{figure}

\section{MAP response LPV modeling and control}\label{sec:Control}
In order to apply the LPV control approach to the MAP regulation problem, we first represent the described system (\ref{eq:MAP response TF}) as an LPV time-delay model. Subsequently, a new time-delayed LPV formulation is developed to design a robust LPV time-delay gain-scheduling controller, where the real-time model parameters are continuously estimated via the MMSRCKF approach and utilized as scheduling parameters. The structure of the closed-loop system with the LPV controller and the real-time MMSRCKF estimator is shown in Fig. \ref{fig:system structure}.  



\subsection{MAP response continuous-time LPV modeling}

\begin{figure}[!t] 
\hspace*{-.17in}
\centering \includegraphics[width=\columnwidth, height=2.15in]{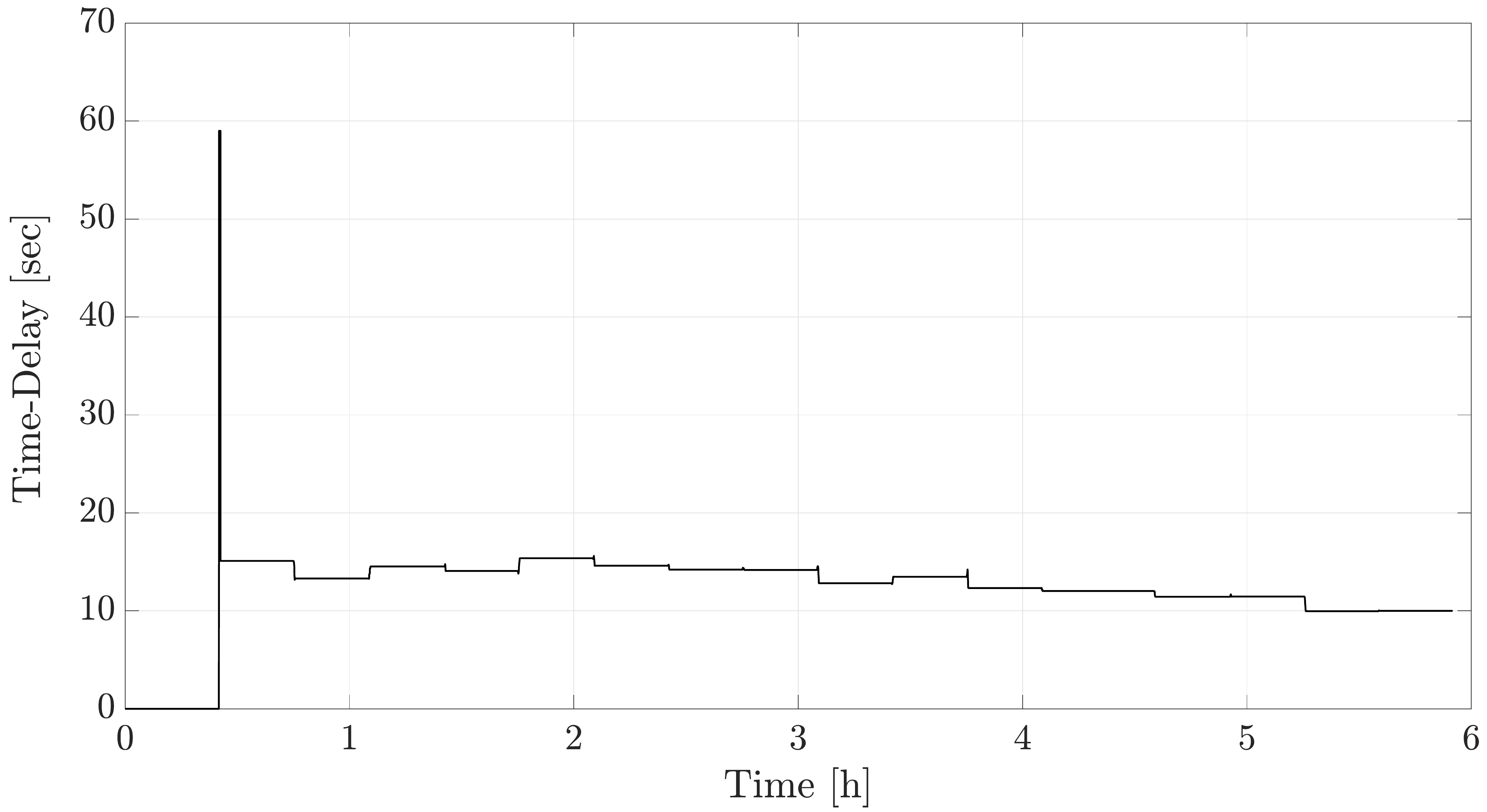} 
\caption{Time-delay parameter estimation} 
\label{fig:Par_Delay}
\end{figure}		

	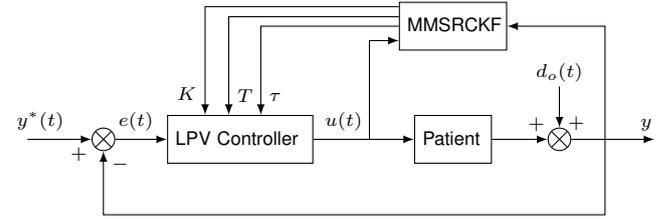
\begin{figure}[!t]
	\centering
		\begin{tikzpicture}[auto, font = {\sf \scriptsize}, cross/.style={path picture={\draw[black] (path picture bounding box.south east) -- (path picture bounding box.north west) (path picture bounding box.south west)-- (path picture bounding box.north east);}}, node distance=2cm,>=latex]
		\node at (-5,0)[input, name=input1] {};
		\node at (-2.2,0)[block,text width=1.7cm,text height=.5em,text depth=0em] (Controller) {{\centering LPV  Controller}};
 	    \node at (.6,0)[block,text width=.8cm,text height=.7em,text depth=.2em] (Patient) {\centering Patient};
 	    \node at (.6,1.5)[block,text width=1.2cm,text height=.5em,text depth=.1em] (ParamEstim) {\centering MMSRCKF};
		\node at (2,.7)[input, name=Disturb] {}; 	    
	    \node at (-4,0) [sum,cross] (sum1) {};
	    \node at (2,0) [sum,cross] (sum2) {};
	    
	    \node [output, left of=Patient, node distance=1.1cm] (output1) {};
	    \node [output, right of=Patient, node distance=1.1cm] (output2) {};
	    \node [output, right of=sum2, node distance=1.25cm] (output3) {};
	    \node [output, right of=sum2, node distance=.6cm] (output4) {};

	    \draw [->] (input1) -- node[name=in1tosum1, pos=0.2] {$y^*(t)$}node[name=ms2, pos=.8, below] {$+$} (sum1);
		\draw [->] (sum1) -- node[name=sum1tosum2, pos=.4] {$e(t)$} (Controller);	 	    
	    \draw [->] (Controller) -- node[name=sum2toCont, pos=0.3] {$u(t)$} (Patient);
	    \draw [->] (Patient) -- node[name=sum2toCont, pos=0.8] {$+$} (sum2); 
	    \draw [->] (sum2) -- node[name=sum2toCont, pos=0.9, above] {$y$} (output3); 
	    \draw [->] (output4) -- +(0,-1) -| node[name=sum2toCont, pos=0.9, right] {$-$} (sum1);
	    \draw [->] (output4) -- +(0,1.5) -- node[name=sum2toCont, pos=0.9, right] {} (ParamEstim.360);
	    
	    \draw [->] (output1) |- node[name=sum2toCont, pos=0.9] {} (ParamEstim.195);

	    \draw [->] (ParamEstim.160) -| node[name=sum5tosum1, pos=0.9, left] {$K$} (Controller.145); 
	    \draw [->] (ParamEstim.170) -| node[name=sum5tosum1, pos=0.9] {$T$} (Controller.115); 
	    \draw [->] (ParamEstim.180) -| node[name=sum5tosum1, pos=0.9] {$\tau$} (Controller.50); 	    
	    \draw [->] (Disturb) -- node[pos=.1, above]{$d_o(t)$} node[name=sum5tosum1, pos=0.9] {$+$} (sum2); 	    
	\end{tikzpicture}
    \caption{Closed-loop system structure}
    \label{fig:system structure}
\end{figure}

By considering the state variable as $x(t) = \Delta MAP (t)$, we can rewrite the state space representation of the first-order time-delayed MAP response model (\ref{eq:MAP response TF}) as follows
\begin{equation}
\begin{matrix}
 \dot{x}(t) & = & -\dfrac{1}{T(t)} x(t) + \dfrac{K(t)}{T(t)} u(t-\tau(t)),\\[0.4cm]
 y(t)& = & x(t) + d_o(t),\:\:\:\:\:\:\:\:\:\:\:\:\:\:\:\:\:\:\:\:\:\:\:\:\:\:\:\:\:\:\:\:\:
\end{matrix}
\label{eq:MAP state space}
\end{equation}
where $y(t)$ is the patient's measured MAP response and $d_o(t)$ denotes output disturbances. In (\ref{eq:MAP state space}), the varying time delay, $\tau (t)$, is appearing in the input signal. In order to utilize the proposed time-delay LPV system control design framework, we need to transform the input delay system into a state-delay LPV representation. To this end, we introduce a filtered input signal $u_a(t)$ as follows
	\beq
	u(s)=\frac{\Omega}{s + \Lambda} u_a (s),
	\eeq
	where $\Omega$ and $\Lambda$ are positive scalars that are selected based on the bandwidth of the actuators. By considering the augmented state vector  $\mathbf{x}_a (t) = [\begin{array}{ccc}
	x(t) & u(t) & x_e(t)\end{array}]^{\text{T}}$, and defining the scheduling parameter vector, $\boldsymbol{\rho}(t) = [\begin{array}{ccc}
	K(t) & T(t) & \tau(t) \end{array}]^{\text{T}}$, the LPV state-delayed state-space representation of the MAP response dynamics takes the following form
\begin{equation}
\!\!\!\!\begin{array}{rl}
 \dot{\mathbf{x}}_a(t)& =  \mathbf{A}(\boldsymbol{\rho}(t)) \mathbf{x}_a(t)+\mathbf{A}_d(\boldsymbol{\rho}(t)) \mathbf{x}_a(t-\tau(t))\\ 
& + \mathbf{B}_1(\boldsymbol{\rho}(t))\mathbf{w}(t)  + \mathbf{B}_2(\boldsymbol{\rho}(t))u(t) \\[0.20cm] 
 y_a(t) & = \mathbf{C}_2(\boldsymbol{\rho}(t)) \mathbf{x}_a(t)+ \mathbf{C}_{2d}(\boldsymbol{\rho}(t)) \mathbf{x}_a(t-\tau(t))\\
 & + \mathbf{D}_{21}(\boldsymbol{\rho}(t)) \mathbf{w}(t),
\end{array}
\label{LPV_MAP}
\end{equation}
where the exogenous disturbance vector $\mathbf{w}(t) = [\begin{array}{cc}
 																												r(t) & d_o(t)\end{array} ]^{\text{T}}$ includes the reference command and output disturbance. The third state $x_e (t)$ is defined for command tracking purposes, \textit{i.e.} $\dot{x}_e (t) = e(t) = r(t)-y(t) = r(t) - (x(t) + d_o(t))$. Thus, the state space matrices of the augmented LPV system (\ref{LPV_MAP}) are obtained as
\begin{align}
			& \mathbf{A}(\boldsymbol{\rho}(t)) =\begin{bmatrix}
- \frac{1}{T(t)} & 0 & 0\\ 
0 & -\Lambda & 0 \\ 
-1 & 0 & 0
\end{bmatrix}, \: \mathbf{A}_d(\boldsymbol{\rho}(t))=\begin{bmatrix}
0 & \frac{K(t)}{T(t)} & 0\\ 
0 & 0 & 0 \\ 
0 & 0 & 0
\end{bmatrix},\nonumber\\
		& \mathbf{B}_1(\boldsymbol{\rho}(t)) =\begin{bmatrix}
0 & 0\\ 
0 & 0\\ 
1 & -1
\end{bmatrix}, \mathbf{B}_2(\boldsymbol{\rho}(t))=\begin{bmatrix}
0\\ 
\Omega\\ 
0
\end{bmatrix},\nonumber\\
			& \mathbf{C}_2(\boldsymbol{\rho}(t)) = \begin{bmatrix}
1 & 0 & 0 \end{bmatrix}, \: \mathbf{D}_{21}(\boldsymbol{\rho}(t))=\begin{bmatrix}
0 & 1
\end{bmatrix},
\label{eq:matrices1}
\end{align}
and $\mathbf{C}_{2d}(\boldsymbol{\rho}(t))$ is a zero matrix with compatible dimensions.

The robust time-delayed LPV control synthesis is examined as next.

\subsection{Robust time-delay LPV control design}\label{sec:control design}
Consider the following state-space representation of an LPV system with a varying state delay
\begin{equation}
\begin{array}{cl}
 \dot{\mathbf{x}}(t)  & =  \mathbf{A}(\boldsymbol{\rho}(t)) \mathbf{x}(t)+\mathbf{A}_d(\boldsymbol{\rho}(t)) \mathbf{x}\big(t-\tau (\boldsymbol{\rho}(t))\big) \\[0.1cm] 
 & + \mathbf{B}_1(\boldsymbol{\rho}(t))\mathbf{w}(t) + \mathbf{B}_2(\boldsymbol{\rho}(t))\mathbf{u}(t) \\[0.25cm] 
  \mathbf{z}(t) & = \mathbf{C}_1(\boldsymbol{\rho}(t)) \mathbf{x}(t) + \mathbf{C}_{1d}(\boldsymbol{\rho}(t)) \mathbf{x}\big(t-\tau (\boldsymbol{\rho}(t))\big)\\[0.1cm] 
  & + \mathbf{D}_{11}(\boldsymbol{\rho}(t)) \mathbf{w}(t)+ \mathbf{D}_{12}(\boldsymbol{\rho}(t)) \mathbf{u}(t)\\[0.25cm]
 \mathbf{y}(t) & = \mathbf{C}_2(\boldsymbol{\rho}(t)) \mathbf{x}(t)+ \mathbf{C}_{2d}(\boldsymbol{\rho}(t)) \mathbf{x}\big(t-\tau (\boldsymbol{\rho}(t))\big)\\[0.1cm]
 & + \mathbf{D}_{21}(\boldsymbol{\rho}(t)) \mathbf{w}(t),\\[0.25cm]
 \mathbf{x}(t_0 + s)  & = \boldsymbol{\phi}(s), \:\:\:\: \forall s \in [-\overline{\tau}, \: \: 0],
\end{array}
\label{LPVsystem}
\end{equation}
where $\mathbf{x}(t) \in \mathbb{R}^n$ is the system state vector, $\mathbf{w}(t) \in \mathbb{R}^{n_w}$ is the vector of exogenous disturbances with finite energy in the space $\mathcal{L}_2[0, \:\: \infty]$, $\mathbf{u}(t) \in \mathbb{R}^{n_u}$ is the input vector, $\mathbf{z}(t) \in \mathbb{R}^{n_z}$ is the vector of outputs to be controlled, $\mathbf{y}(t) \in \mathbb{R}^{n_y}$ is the vector of measurable outputs, $\boldsymbol{\phi}(s) \in \mathcal{C}([-\overline{\tau} \:\: 0], \mathbb{R}^n)$ is the system initial condition, and the state space matrices in (\ref{LPVsystem}), \textit{i.e.} $\mathbf{A}(\cdot)$, $\mathbf{A}_d(\cdot)$, $\mathbf{B}_1(\cdot)$, $\mathbf{B}_2(\cdot)$, $\mathbf{C}_1(\cdot)$, $\mathbf{C}_{1d}(\cdot)$, $\mathbf{D}_{11}(\cdot)$, $\mathbf{D}_{12}(\cdot)$,  $\mathbf{C}_2(\cdot)$, $\mathbf{C}_{2d}(\cdot)$, and  $\mathbf{D}_{21}(\cdot)$ are real-valued matrices which are  continuous functions of the time-varying parameter vector $\boldsymbol{\rho}(\cdot) \in \mathscr{F}^\nu _\mathscr{P}$. The scheduling parameter vector is assumed to be measurable in real-time and the set $\mathscr{F}^\nu _\mathscr{P}$ denotes the set of allowable scheduling parameter trajectories defined as
\begin{align}
 \mathscr{F}^\nu _\mathscr{P} \triangleq \{\boldsymbol{\rho}(t) \in \mathcal{C}(\mathbb{R}_{+},\mathbb{R}^{s}):\boldsymbol{\rho}(t) \in \mathscr{P}, |\dot{\rho}_i (t)| \leq \nu_i,\nonumber\\  i=1,2,\dots,n_s, \: \forall t \in \mathbb{R}_{\geq 0} \},
 \label{eq:parametertraj}
\end{align}
where $n_s$ is the number of parameters and  $\mathscr{P}$ is a compact subset of $\mathbb{R}^{n_s}$. Also, $\tau (\boldsymbol{\rho}(t))$ is a differentiable scalar function representing the parameter-varying time delay, that is considered to be dependent on the scheduling parameter vector and lies in the set $\mathscr{T}^\mu$ defined as  
\begin{align}
\mathscr{\mathscr{T}^\mu} \triangleq \{ \tau (\boldsymbol{\rho}(t)) \in \mathcal{C}(\mathscr{P},\mathbb{R}_{\geq 0}) : 0 \leq \tau (\cdot) \leq \overline{\tau} < \infty, \nonumber\\ \dot{\tau}(\cdot) \leq \mu, \: \forall t \in \mathbb{R}_{\geq 0}\}.
\end{align}
 Since the delay is considered to be dependent on the scheduling parameter vector $\boldsymbol{\rho}(t)$, as a result, the delay bound should be incorporated into the parameter set $\mathscr{F}^\nu _\mathscr{P}$. 

In the present work, Lyapunov-Krasovskii functionals are utilized to obtain less conservative results, which are valid for bounded parameter variation rates \cite{apkarian1998advanced}. We seek a gain-scheduling LPV controller to meet the following objectives:
	\begin{itemize}
		\item Input-to-state stability (ISS) of the closed-loop system in the presence of parameter and delay variations, uncertainties, and disturbances, and
		\item Minimization of the worst case amplification of the induced $\mathcal{L}_2$-norm of the mapping from the disturbances $\mathbf{w}(t)$ to the controlled output $\mathbf{z}(t)$, given by
			\beq
				\Vert \mathbf{T}_{\mathbf{z}\mathbf{w}}\Vert_{i,2} = \underset{\boldsymbol{\rho}(t) \in \mathscr{F}^\nu _\mathscr{P}}{\sup} \:\:\: \underset{\Vert \mathbf{w}(t) \Vert_2 \neq 0}{\sup}\:\: \frac{\Vert \mathbf{z}(t) \Vert_2}{\Vert \mathbf{w}(t) \Vert_2}.
				\label{eq:Performance Index}
			\eeq
	\end{itemize}

\noindent Accordingly, a full-order dynamic output-feedback controller in the following form is considered:
\begin{equation}
\begin{array}{cl}
\dot{\mathbf{x}}_k(t) & = \mathbf{A}_k (\boldsymbol{\rho}) \mathbf{x}_k(t)+ \mathbf{A}_{dk}(\rho)\mathbf{x}_k(t-\tau(t))+\mathbf{B}_k(\boldsymbol{\rho})\mathbf{y}(t),\\ 
\mathbf{u}(t) & = \mathbf{C}_k (\boldsymbol{\rho}) \mathbf{x}_k(t)+ \mathbf{C}_{dk}(\boldsymbol{\rho})\mathbf{x}_k(t-\tau(t))+\mathbf{D}_k(\boldsymbol{\rho})\mathbf{y}(t),
\end{array}
\label{controller}
\end{equation}
where $\mathbf{x}_k (t) \in \mathbb{R}^{n}$ is the controller state vector and $\mathbf{x}_k(t-\tau(t)) \in \mathbb{R}^{n}$ is the delayed state of the controller. Considering the system dynamics (\ref{LPVsystem}) and the controller (\ref{controller}), the closed-loop system would be as follows:
\beq
\begin{array}{cl}
\dot{\mathbf{x}}_{cl}(t)\!\!\!\!  & =  \mathbf{A}_{cl}\: {\mathbf{x}}_{cl}(t) + \mathbf{A}_{d,cl}\: {\mathbf{x}}_{cl}(t-\tau(t)) + \mathbf{B}_{cl}\:\mathbf{w}(t),\\ 
\mathbf{z}(t)\!\!\!\! & =  \mathbf{C}_{cl}\:{\mathbf{x}}_{cl}(t) +  \mathbf{C}_{d,cl}\:{\mathbf{x}}_{cl}(t-\tau(t)) + \mathbf{D}_{cl}\: \mathbf{w}(t),
\end{array}
\label{eq:closed-loop system}
\eeq
where
\begin{align*}
		& \mathbf{A}_{cl}=\begin{bmatrix}
\mathbf{A} + \mathbf{B}_2 \mathbf{D}_k \mathbf{C}_2 & \mathbf{B}_2 \mathbf{C}_k\\ 
\mathbf{B}_k \mathbf{C}_2 & \mathbf{A}_k 
\end{bmatrix}, \\
& \mathbf{A}_{d,cl}\!=\!\begin{bmatrix}
\mathbf{A}_d + \mathbf{B}_2 \mathbf{D}_k \mathbf{C}_{2d} & \mathbf{B}_2 \mathbf{C}_{dk}\\ 
\mathbf{B}_k \mathbf{C}_{2d} & \mathbf{A}_{dk} 
\end{bmatrix},  \mathbf{B}_{cl}=\begin{bmatrix}
\mathbf{B}_1 + \mathbf{B}_2 \mathbf{D}_k \mathbf{D}_{21} \\
\mathbf{B}_k \mathbf{D}_{21} 
\end{bmatrix}, \\
& \mathbf{C}_{cl}=\begin{bmatrix}
\mathbf{C}_1 + \mathbf{D}_{12} \mathbf{D}_k \mathbf{C}_2 & \mathbf{D}_{12} \mathbf{C}_k
\end{bmatrix}, \\
& \mathbf{C}_{d,cl} =\begin{bmatrix}
\mathbf{C}_{1d} + \mathbf{D}_{12} \mathbf{D}_k \mathbf{C}_{2d} & \mathbf{D}_{12} \mathbf{C}_{dk}
\end{bmatrix}, \\
& \mathbf{D}_{cl}= \mathbf{D}_{11} + \mathbf{D}_{12} \mathbf{D}_k \mathbf{D}_{21}, 
\end{align*}
\noindent and $\mathbf{x}_{cl}(t) = [\begin{array}{cc}
\mathbf{x}(t) &\mathbf{x}_k (t)\end{array} ]^{\text{T}}$, and the dependence on the scheduling parameter has been dropped for clarity. Now, considering the closed-loop system (\ref{eq:closed-loop system}), the following result provides sufficient conditions for the synthesis of a delayed output-feedback controller which guarantees closed-loop asymptotic stability and a specified level of disturbance rejection performance as defined in (\ref{eq:Performance Index}).

\begin{theorem}\label{thm:thm3} The system (\ref{LPVsystem}) is asymptotically stable for parameters $\boldsymbol{\rho}(t) \in \mathscr{F}^\nu _\mathscr{P}$ and all delays $\tau (t) \in \mathscr{T}^\mu$ and satisfy the condition $||\mathbf{z}(t)||_2 \leq \gamma ||\mathbf{w}(t)||_2$ for the closed-loop system (\ref{eq:closed-loop system}), if there exists a continuously differentiable matrix function $\widetilde{\mathbf{P}} : \mathbb{R}^{s}\rightarrow\mathbb{S}^{2n}_{++}$, parameter dependent matrix functions $\mathbf{X}, \mathbf{Y}  : \mathbb{R}^{s}\rightarrow\mathbb{S}^{n}_{++}$, constant matrices $\widetilde{\mathbf{Q}}$, $\widetilde{\mathbf{R}} \in \mathbb{S}^{n}_{++}$, parameter dependent matrices $\widehat{A}$, $\widehat{A}_d$, $\widehat{B}$, $\widehat{C}$, $\widehat{C}_d$, $\widehat{D}_k$, and scalars $\mathbf{\gamma} > 0$, and $\lambda_2$, $\lambda_3$  such that the following LMI conditions hold
\begin{equation}
\begin{array}{l}
\left[\begin{array}{ccc}
-2\widetilde{\mathbf{V}} & \widetilde{\mathbf{P}}-\lambda_2\widetilde{\mathbf{V}} + \mathscr{A} & -\lambda_3\widetilde{\mathbf{V}} + \mathscr{A}_d\\
\star & \widetilde{\mathbf{\Psi}}_{22} + \lambda_2(\mathscr{A}+\mathscr{A}^{\text{T}}) &  \widetilde{\mathbf{R}} + \lambda_3 \mathscr{A}^{\text{T}} + \lambda_2 \mathscr{A}_d \\
\star & \star & \widetilde{\mathbf{\Xi}}_{22} + \lambda_3(\mathscr{A}_d+\mathscr{A}_d^{\text{T}}) \\
\star & \star & \star \\
\star & \star & \star \\
\star & \star & \star  
\end{array}\right.\\
\qquad\qquad\quad\quad\quad\qquad\left.\begin{array}{ccc}
\mathscr{B} & \mathbf{0} & \widetilde{\mathbf{V}}+\overline{\tau}\widetilde{\mathbf{R}} \\
\lambda_2 \mathscr{B}  &  \mathscr{C}^{\text{T}} & \lambda_2\widetilde{\mathbf{V}} - \widetilde{\mathbf{P}}\\
\lambda_3 \mathscr{B} & \mathscr{C}_d^{\text{T}}  & \lambda_3\widetilde{\mathbf{V}}\\ 
-\gamma \mathbf{I} & \mathscr{D}^{\text{T}} & \mathbf{0}\\
\star & -\gamma \mathbf{I} & \mathbf{0}\\
\star & \star & (-1 -2\overline{\tau})\widetilde{\mathbf{R}} 
\end{array}\right]\prec\mathbf{0},
\end{array}
\label{eq:LMI closed-loop}
\end{equation}
with
\begin{equation}
\begin{array}{lll}
\widetilde{\mathbf{V}}  & = & \begin{bmatrix}
\mathbf{Y} & \mathbf{I}\\ 
\mathbf{I} & \mathbf{X}
\end{bmatrix},\\[0.3cm]
\mathscr{A} & = &\begin{bmatrix}
\mathbf{A} \mathbf{Y}+\mathbf{B}_2\widehat{C} & \mathbf{A}+\mathbf{B}_2 \mathbf{D}_k \mathbf{C}_2\\ 
\widehat{A} & \mathbf{X}\mathbf{A}+\widehat{B}\mathbf{C}_2
\end{bmatrix} ,\\[0.4cm]
\mathscr{A}_d & = & \begin{bmatrix}
\mathbf{A}_d\mathbf{Y}+\mathbf{B}_2\widehat{C}_d & \mathbf{A}_d+\mathbf{B}_2 \mathbf{D}_k \mathbf{C}_{2d}\\ 
\widehat{A}_d & \mathbf{X}\mathbf{A}_d+\widehat{B}\mathbf{C}_{2d}
\end{bmatrix},\\[0.4cm]
\mathscr{B} & = & \begin{bmatrix}
\mathbf{B}_1 + \mathbf{B}_2 \mathbf{D}_k \mathbf{D}_{21}\\ 
\mathbf{X}\mathbf{B}_1 + \widehat{B} \mathbf{D}_{21}
\end{bmatrix},\\[0.3cm] 
\mathscr{C} & = & \begin{bmatrix}
\mathbf{C}_1 \mathbf{Y}+ \mathbf{D}_{12} \widehat{C} & \mathbf{C}_1 + \mathbf{D}_{12} \mathbf{D}_k \mathbf{C}_2 
\end{bmatrix},\\[0.3cm]
\mathscr{C}_d & = & \begin{bmatrix}
\mathbf{C}_{1d} \mathbf{Y}+ \mathbf{D}_{12} \widehat{C}_d & \mathbf{C}_{1d} + \mathbf{D}_{12} \mathbf{D}_k \mathbf{C}_{2d} 
\end{bmatrix},\\[0.3cm]
\mathscr{D} & = & \begin{bmatrix}
\mathbf{D}_{11} + \mathbf{D}_{12} \mathbf{D}_k \mathbf{D}_{21}
\end{bmatrix},\\[0.3cm] \widetilde{\mathbf{\Psi}}_{22} & = & \bigg[ \sum_{i=1}^s \pm \Big(\nu_i \frac{\partial \widetilde{\mathbf{P}}(\boldsymbol{\rho})}{\partial \rho_i}\Big) \bigg] + \widetilde{\mathbf{Q}}-\widetilde{\mathbf{R}}  ,\\[0.3cm]
\widetilde{\mathbf{\Xi}}_{22} & = &-\bigg[1- \sum_{i=1}^s \pm \Big(\nu_i \frac{\partial \tau}{\partial \rho_i}\Big) \bigg] \widetilde{\mathbf{Q}}-\widetilde{\mathbf{R}} .
\end{array}
\label{Eqcoef1}
\end{equation}
\end{theorem}

\begin{proof}
Refer to \cite{Briat2015}.
\end{proof} 

For robust LPV control synthesis, we consider the class of uncertain time-delay LPV systems with the norm-bounded uncertainties in the state and delayed state matrices as:
\begin{equation}
\begin{array}{cl}
 \dot{\mathbf{x}}(t) & =  \mathbf{A}_{\Delta}(\boldsymbol{\rho}(t)) \mathbf{x}(t)+\mathbf{A}_{\Delta d}(\boldsymbol{\rho}(t)) \mathbf{x}(t-\tau(t))\\ 
 & + \mathbf{B}_1(\boldsymbol{\rho}(t))\mathbf{w}(t) + \mathbf{B}_2(\boldsymbol{\rho}(t))\mathbf{u}(t) \\[0.25cm] 
  \mathbf{z}(t) & = \mathbf{C}_1(\boldsymbol{\rho}(t)) \mathbf{x}(t) + \mathbf{C}_{1d}(\boldsymbol{\rho}(t)) \mathbf{x}(t-\tau(t))\\ 
 &  + \mathbf{D}_{11}(\boldsymbol{\rho}(t)) \mathbf{w}(t) + \mathbf{D}_{12}(\boldsymbol{\rho}(t)) \mathbf{u}(t)\\[0.25cm]
 y(t)& = \mathbf{C}_2(\boldsymbol{\rho}(t)) \mathbf{x}(t)+ \mathbf{C}_{2d}(\boldsymbol{\rho}(t)) \mathbf{x}(t-\tau(t))\\
 & + \mathbf{D}_{21}(\boldsymbol{\rho}(t)) \mathbf{w}(t),\\[0.25cm]
  \mathbf{x}(t_0 + s)  & = \boldsymbol{\phi}(s), \:\:\:\: \forall s \in [-\overline{\tau}, \: \: 0],
\end{array}
\label{unc_LPVsystem}
\end{equation}
\noindent where $\mathbf{A}_{\Delta}(\boldsymbol{\rho}(t)) = \mathbf{A}(\boldsymbol{\rho}(t))  + \boldsymbol{\Delta} \mathbf{A}(t)$, $\mathbf{A}_{\Delta d}(\boldsymbol{\rho}(t)) = \mathbf{A}_d(\boldsymbol{\rho}(t))  + \boldsymbol{\Delta} \mathbf{A}_d(t)$ are bounded matrices containing parametric uncertainties. The norm-bounded uncertainties are assumed to satisfy the following relations
\noindent
\begin{equation}
\left[\begin{array}{c}
\boldsymbol{\Delta} \mathbf{A}(t)\\
\boldsymbol{\Delta} \mathbf{A}_d(t) 
\end{array} 
\right] = \mathbf{H} \boldsymbol{\Delta} (t) \left[ \begin{array}{c}
\mathbf{E}_1\\
\mathbf{E}_2\end{array} \right],
\label{uncertainties}
 \end{equation} 
\noindent where $\mathbf{H} \in \mathbb{R}^{n \times i} $, $\mathbf{E}_1 \in \mathbb{R}^{j \times n} $, $\mathbf{E}_2 \in \mathbb{R}^{j \times n} $ are known constant matrices and $\boldsymbol{\Delta} (t) \in \mathbb{R}^{i \times j}$ is an unknown  time-varying uncertainty matrix function satisfying
\begin{equation}
 \boldsymbol{\Delta}^T (t) \boldsymbol{\Delta} (t)  \leq \mathbf{I}.
\label{Delta1}
\end{equation}

Considering the uncertain time-delayed LPV system (\ref{unc_LPVsystem}), the following result provides sufficient conditions for the synthesis of a robust time-delayed output-feedback LPV controller which guarantees the asymptotic stability and a specified level of disturbance rejection performance as in (\ref{eq:Performance Index}) for the uncertain closed-loop time-delay system.
\\
\begin{theorem}\label{thm:thm4} There exists a full-order robust output-feedback LPV controller of the form (\ref{controller}) which first, asymptotically stabilizes the uncertain LPV system (\ref{unc_LPVsystem}) with all admissible uncertainties $\boldsymbol{\Delta} \mathbf{A}(t)$ and $\boldsymbol{\Delta} \mathbf{A}_d(t)$ of the form (\ref{uncertainties}) and all $\boldsymbol{\Delta}(t)$ satisfying (\ref{Delta1}) with $\tau (t) \in \mathscr{T}^\mu$ and $\boldsymbol{\rho}(t) \in \mathscr{F}^\nu _\mathscr{P}$ and second, satisfies the condition $||\mathbf{z}(t)||_2 \leq \gamma ||\mathbf{w}(t)||_2$ for the closed-loop system, if there exists a continuously differentiable matrix function $\widetilde{\mathbf{P}} : \mathbb{R}^{s}\rightarrow\mathbb{S}^{2n}_{++}$, parameter dependent matrix functions $\mathbf{X}, \mathbf{Y}  : \mathbb{R}^{s}\rightarrow\mathbb{S}^{n}_{++}$, constant matrices $\widetilde{\mathbf{Q}}$, $\widetilde{\mathbf{R}} \in \mathbb{S}^{n}_{++}$, parameter dependent matrices $\widehat{A}$, $\widehat{A}_d$, $\widehat{B}$, $\widehat{C}$, $\widehat{C}_d$, $\widehat{D}_k$, and scalars $\mathbf{\gamma}> 0$,  $\mathbf{\epsilon} > 0$, and $\lambda_2$, $\lambda_3$ such that the following LMI is feasible.
\begin{small}
\begin{equation}
\hspace{-.12in}
\begin{array}{l}
\left[\begin{array}{cccc}
-2\widetilde{\mathbf{V}} & \widetilde{\mathbf{P}}-\lambda_2\widetilde{\mathbf{V}} + \mathscr{A}  & -\lambda_3\widetilde{\mathbf{V}} + \mathscr{A}_d  & \mathscr{B}\\[2pt]
\star & \widetilde{\mathbf{\Psi}}_{22} + \lambda_2(\mathscr{A}+\mathscr{A}^{\text{T}})  &  \widetilde{\mathbf{R}} + \lambda_3 \mathscr{A}^{\text{T}} + \lambda_2 \mathscr{A}_d  & \lambda_2 \mathscr{B}\\[2pt]
\star & \star & \widetilde{\mathbf{\Xi}}_{22} + \lambda_3(\mathscr{A}_d+\mathscr{A}_d^{\text{T}})  & \lambda_3 \mathscr{B}\\[2pt]
\star & \star & \star & -\gamma \mathbf{I}  \\
\star & \star & \star & \star  \\
\star & \star & \star  & \star  \\
\star & \star & \star  & \star  \\
\star & \star & \star  & \star  
\end{array}\right.\\ [15pt]
\!\!\!\left.\begin{array}{cccc}
\mathbf{0} & \widetilde{\mathbf{V}}+\overline{\tau}\widetilde{\mathbf{R}} 
& \quad \: \left[\begin{array}{cc}
\mathbf{H}^{\text{T}}  & \mathbf{H}^{\text{T}} \mathbf{X}\\
\mathbf{0}  & \mathbf{0} 
\end{array}\right] & \mathbf{0} \\ [9pt]
\mathscr{C}^{\text{T}} & \lambda_2\widetilde{\mathbf{V}} - \widetilde{\mathbf{P}} & \lambda_2 \left[\begin{array}{cc}
\mathbf{H}^{\text{T}}  & \mathbf{H}^{\text{T}} \mathbf{X}\\
\mathbf{0}  & \mathbf{0} 
\end{array}\right] & \mathbf{\epsilon} \left[\begin{array}{cc}\mathbf{Y}^{\text{T}} \mathbf{E}_1^{\text{T}} & \mathbf{0} \\
                                  \mathbf{E}_1 ^{\text{T}} & \mathbf{0}\end{array}\right]\\[9pt]
\mathscr{C}_d^{\text{T}}  & \lambda_3\widetilde{\mathbf{V}} & \lambda_3 \left[\begin{array}{cc}
\mathbf{H}^{\text{T}}  & \mathbf{H}^{\text{T}} \mathbf{X}\\
\mathbf{0}  & \mathbf{0} 
\end{array}\right] & \mathbf{\epsilon} \left[\begin{array}{cc}\mathbf{Y}^{\text{T}} \mathbf{E}_2^{\text{T}} & \mathbf{0} \\
                                  \mathbf{E}_2 ^{\text{T}} & \mathbf{0}\end{array}\right]\\ [9pt]
\mathscr{D}^{\text{T}} & \mathbf{0} & \mathbf{0} & \mathbf{0}\\[2pt]
-\gamma \mathbf{I} & \mathbf{0} & \mathbf{0} & \mathbf{0}\\[2pt]
\star & (-1 -2\overline{\tau})\widetilde{\mathbf{R}}& \mathbf{0} & \mathbf{0} \\[2pt]
 \star & \star & -\mathbf{\epsilon} \mathbf{I} & \mathbf{0}\\[2pt]
 \star & \star & \star & -\mathbf{\epsilon} \mathbf{I}
\end{array}\right]\!\!\! \prec\!\mathbf{0},
\end{array}
\label{eq:LMIunc}
\end{equation}

\end{small}
\noindent with $\widetilde{\mathbf{V}}$, $\mathscr{A}$, $\mathscr{A}_d$, $\mathscr{B}$, $\mathscr{C}$, $\mathscr{C}_d$, $\mathscr{D}$, $\widetilde{\mathbf{\Psi}}_{22}$, and $\widetilde{\mathbf{\Xi}}_{22}$ as in (\ref{Eqcoef1}).
\end{theorem}

\begin{proof} By substituting the norm-bounded matrices $\mathbf{A}_{\Delta}(\boldsymbol{\rho}(t)) = \mathbf{A}(\boldsymbol{\rho}(t))  +\boldsymbol{\Delta} \mathbf{A}(t)$ and $\mathbf{A}_{\Delta d}(\boldsymbol{\rho}(t)) = \mathbf{A}_d(\boldsymbol{\rho}(t))  + \boldsymbol{\Delta} \mathbf{A}_d(t)$ containing parametric uncertainties for $\mathbf{A}(\boldsymbol{\rho}(t))$ and $\mathbf{A}_d(\boldsymbol{\rho}(t))$ into the LMI condition (\ref{eq:LMI closed-loop}) of Theorem \ref{thm:thm3}, we obtain a new LMI condition (\ref{eq:LMIunc2}), which can be written as summation of the initial LMI constraint (\ref{eq:LMI closed-loop}) and the LMI corresponding to the uncertain parts as shown in (\ref{eq:LMIunc3}).

\begin{figure*}[!t]
\begin{small}
\begin{equation}
\begin{array}{l}
\left[\begin{array}{cc}
-2\widetilde{\mathbf{V}} & \widetilde{\mathbf{P}}-\lambda_2\widetilde{\mathbf{V}} + \mathscr{A} + \left[\begin{array}{cc}\boldsymbol{\Delta} \mathbf{A} \mathbf{Y} & \boldsymbol{\Delta} \mathbf{A} \\
                                  \mathbf{0} & \mathbf{X} \boldsymbol{\Delta} \mathbf{A}\end{array}\right] \\ [2pt] 
\star & \widetilde{\mathbf{\Psi}}_{22} + \lambda_2(\mathscr{A}+\mathscr{A}^{\text{T}}) + \lambda_2(\left[\begin{array}{cc}\boldsymbol{\Delta} \mathbf{A} \mathbf{Y} & \boldsymbol{\Delta} \mathbf{A} \\
                                  \mathbf{0} & \mathbf{X} \boldsymbol{\Delta} \mathbf{A}\end{array}\right]+\left[\begin{array}{cc}\boldsymbol{\Delta} \mathbf{A} \mathbf{Y} & \boldsymbol{\Delta} \mathbf{A} \\
                                  \mathbf{0} & \mathbf{X} \boldsymbol{\Delta} \mathbf{A}\end{array}\right]^{\text{T}})\\ 
\star & \star\\
\star & \star\\
\star & \star\\
\star & \star
\end{array}\right.\\ [60pt]
\quad \quad \left.\begin{array}{cccc}
-\lambda_3\widetilde{\mathbf{V}} + \mathscr{A}_d + \left[\begin{array}{cc}\boldsymbol{\Delta}  \mathbf{A}_d \mathbf{Y} & \boldsymbol{\Delta}  \mathbf{A}_d \\
                                  \mathbf{0} & \mathbf{X} \boldsymbol{\Delta}  \mathbf{A}_d\end{array}\right] & \mathscr{B} & \mathbf{0} & \widetilde{\mathbf{V}}+\overline{\tau}\widetilde{\mathbf{R}} \\ [10pt] 
\widetilde{\mathbf{R}} + \lambda_3 \mathscr{A}^{\text{T}} + \lambda_2 \mathscr{A}_d + \lambda_3 \left[\begin{array}{cc}\boldsymbol{\Delta} \mathbf{A} \mathbf{Y} & \boldsymbol{\Delta} \mathbf{A} \\
                                  \mathbf{0} & \mathbf{X} \boldsymbol{\Delta} \mathbf{A}\end{array}\right]^{\text{T}} + \lambda_2  \left[\begin{array}{cc}\boldsymbol{\Delta}  \mathbf{A}_d \mathbf{Y} & \boldsymbol{\Delta}  \mathbf{A}_d \\
                                  \mathbf{0} & \mathbf{X} \boldsymbol{\Delta}  \mathbf{A}_d\end{array}\right] & \lambda_2 \mathscr{B}  &  \mathscr{C}^{\text{T}} & \lambda_2\widetilde{\mathbf{V}} - \widetilde{\mathbf{P}}\\ [10pt] 
\widetilde{\mathbf{\Xi}}_{22} + \lambda_3(\mathscr{A}_d+\mathscr{A}_d^{\text{T}}) + \lambda_3(\left[\begin{array}{cc}\boldsymbol{\Delta}  \mathbf{A}_d \mathbf{Y} & \boldsymbol{\Delta}  \mathbf{A}_d \\
                                  \mathbf{0} & \mathbf{X} \boldsymbol{\Delta}  \mathbf{A}_d\end{array}\right] + \left[\begin{array}{cc}\boldsymbol{\Delta}  \mathbf{A}_d \mathbf{Y} & \boldsymbol{\Delta}  \mathbf{A}_d \\
                                  \mathbf{0} & \mathbf{X} \boldsymbol{\Delta}  \mathbf{A}_d\end{array}\right]^{\text{T}}) & \lambda_3 \mathscr{B} & \mathscr{C}_d^{\text{T}}  & \lambda_3\widetilde{\mathbf{V}}\\ [2pt] 
\star & -\gamma \mathbf{I} & \mathscr{D}^{\text{T}} & \mathbf{0}\\ [2pt]
\star & \star & -\gamma \mathbf{I} & \mathbf{0}\\ [2pt]
\star & (-1 -2\overline{\tau})\widetilde{\mathbf{R}}& \mathbf{0} & \mathbf{0} \\ [2pt]
 \star & \star & -\mathbf{\epsilon} \mathbf{I} & \mathbf{0}\\ [2pt]
 \star  & \star & \star & (-1 -2\overline{\tau})\widetilde{\mathbf{R}}
\end{array}\right]\prec\mathbf{0},
\end{array}
\label{eq:LMIunc2}
\end{equation} 
\end{small}
\end{figure*}

\begin{footnotesize}
\begin{equation}
\!\!\! \begin{array}{l}
(\ref{eq:LMIunc2}) = (\ref{eq:LMI closed-loop}) + \\ [10pt] 
\begin{array}{l}
\left[\begin{array}{cc}
\mathbf{0} &  \left[\begin{array}{cc}\boldsymbol{\Delta} \mathbf{A} \mathbf{Y} & \boldsymbol{\Delta} \mathbf{A} \\
                                  \mathbf{0} & \mathbf{X} \boldsymbol{\Delta} \mathbf{A}\end{array}\right]\\[6pt]  
\star &  \lambda_2(\left[\begin{array}{cc}\boldsymbol{\Delta} \mathbf{A} \mathbf{Y} & \boldsymbol{\Delta} \mathbf{A} \\
                                  \mathbf{0} & \mathbf{X} \boldsymbol{\Delta} \mathbf{A}\end{array}\right]+\left[\begin{array}{cc}\boldsymbol{\Delta} \mathbf{A} \mathbf{Y} & \boldsymbol{\Delta} \mathbf{A} \\
                                  \mathbf{0} & \mathbf{X} \boldsymbol{\Delta} \mathbf{A}\end{array}\right]^{\text{T}}) \\ 
\star & \star\\
\star & \star\\
\star & \star\\
\star & \star
\end{array}\right.\\[40pt]
\!\!\!\!\!\!\!\!\!\left.\begin{array}{cccc}
\left[\begin{array}{cc}\boldsymbol{\Delta}  \mathbf{A}_d \mathbf{Y} & \boldsymbol{\Delta}  \mathbf{A}_d \\
                                  \mathbf{0} & \mathbf{X} \boldsymbol{\Delta}  \mathbf{A}_d\end{array}\right] & \mathbf{0} & \mathbf{0} & \mathbf{0} \\ [10pt]
\lambda_3 \left[\begin{array}{cc}\boldsymbol{\Delta} \mathbf{A} \mathbf{Y} & \boldsymbol{\Delta} \mathbf{A} \\
                                  \mathbf{0} & \mathbf{X} \boldsymbol{\Delta} \mathbf{A}\end{array}\right]^{\text{T}} \!\! \!\!+ \!\! \lambda_2 \!\! \left[\begin{array}{cc}\boldsymbol{\Delta}  \mathbf{A}_d \mathbf{Y} & \boldsymbol{\Delta}  \mathbf{A}_d \\
                                  \mathbf{0} & \mathbf{X} \boldsymbol{\Delta}  \mathbf{A}_d\end{array}\right] & \mathbf{0}  &  \mathbf{0} & \mathbf{0}\\ [10pt]
\lambda_3(\left[\begin{array}{cc}\boldsymbol{\Delta}  \mathbf{A}_d \mathbf{Y} & \boldsymbol{\Delta}  \mathbf{A}_d \\
                                  \mathbf{0} & \mathbf{X} \boldsymbol{\Delta}  \mathbf{A}_d\end{array}\right] \!\! + \!\! \left[\begin{array}{cc}\boldsymbol{\Delta}  \mathbf{A}_d \mathbf{Y} & \boldsymbol{\Delta}  \mathbf{A}_d \\
                                  \mathbf{0} & \mathbf{X} \boldsymbol{\Delta}  \mathbf{A}_d\end{array}\right]^{\text{T}}) & \mathbf{0} & \mathbf{0}  & \mathbf{0}\\  
\star & \mathbf{0} & \mathbf{0} & \mathbf{0}\\  
\star & \star & \mathbf{0} & \mathbf{0}\\ 
\star & \star & \star & \mathbf{0} 
\end{array}\right]\!\!\!\prec\!\!\mathbf{0},
\end{array}
\end{array}
\label{eq:LMIunc3}
\end{equation}
\end{footnotesize}
\noindent This condition can equivalently be written as
\begin{footnotesize}
\begin{equation}
\begin{array}{l}
(\ref{eq:LMIunc2}) = (\ref{eq:LMI closed-loop}) + \\ [10pt] 
\mathbf{He}\Bigg(\begin{array}{l}
\left[\begin{array}{cccccc}
\: \quad \left[\begin{array}{cc}\mathbf{H} & \mathbf{0} \\
                              \mathbf{X}\mathbf{H} & \mathbf{0}\end{array}\right] \\ [9pt]
\lambda_2 \left[\begin{array}{cc}\mathbf{H} & \mathbf{0} \\
                                  \mathbf{X}\mathbf{H} & \mathbf{0}\end{array}\right] \\ [9pt]
\lambda_3 \left[\begin{array}{cc}\mathbf{H} & \mathbf{0} \\
                                  \mathbf{X}\mathbf{H} & \mathbf{0}\end{array}\right] \\
\mathbf{0}\\
\mathbf{0}\\
\mathbf{0}\end{array}\right]
\left[\begin{array}{cc}
\boldsymbol{\Delta}(t) & \mathbf{0}\\
\mathbf{0}  & \boldsymbol{\Delta} (t)
\end{array}\right]\\[60pt]
\hspace{-.22in}\left[\begin{array}{cccccc}
\mathbf{0}, & \!\! \left[\begin{array}{cc}
\mathbf{E}_1 \mathbf{Y} & \mathbf{E}_1\\
\mathbf{0}  & \mathbf{0} 
\end{array}\right], & \left[\begin{array}{cc}
\mathbf{E}_2 \mathbf{Y} & \mathbf{E}_2\\
\mathbf{0}  & \mathbf{0} 
\end{array}\right], & \mathbf{0}, & \mathbf{0}, & \mathbf{0}
\end{array}\right]\Bigg)\!\!\prec\!\mathbf{0}.
\end{array}
\end{array}
\end{equation}
\end{footnotesize}
\noindent Finally, by using the following inequality \cite{xie1996output} 
\begin{equation}
\boldsymbol{\Theta} \boldsymbol{\Delta} (t) \boldsymbol{\Phi} + \boldsymbol{\Phi}^{\text{T}} \boldsymbol{\Delta}^{\text{T}} (t) \boldsymbol{\Theta}^{\text{T}} \leq \mathbf{\epsilon}^{-1} \boldsymbol{\Theta} \boldsymbol{\Theta}^{\text{T}} + \mathbf{\epsilon} \boldsymbol{\Phi} ^{\text{T}} \boldsymbol{\Phi},
\label{eq:ineq}
\end{equation}
\noindent which holds for all scalars $\mathbf{\epsilon} > 0$ and all constant matrices $\boldsymbol{\Theta}$ and $\boldsymbol{\Phi}$ of appropirate dimensions, and using the Schur complement \cite{boyd1994linear}, the final LMI condition (\ref{eq:LMIunc}) is obtained. 
\end{proof} 

Once the parameter dependent matrices $\mathbf{X}$, $\mathbf{Y}$, $\widehat{A}$, $\widehat{A}_d$, $\widehat{B}$, $\widehat{C}$, $\widehat{C}_d$, and $\widehat{D}_k$ satisfying the  LMI condition (\ref{eq:LMIunc}) are obtained, the delayed output-feedback controller matrices can be computed as follows:

\noindent 1. Determine $\mathbf{M}$ and $\mathbf{N}$ from the factorization problem
\begin{equation}
\mathbf{I} - \mathbf{X}\mathbf{Y} = \mathbf{N} \mathbf{M}^{\text{T}},
\end{equation} 
\noindent where the obtained $\mathbf{M}$ and $\mathbf{N}$ matrices are square and invertible in the case of a full-order controller.

\noindent 2. Compute the following parameter matrices:
\begin{equation}
\begin{array}{cl}
\widehat{A}&= \mathbf{X} \mathbf{A} \mathbf{Y} + \mathbf{X} \mathbf{B}_2 \mathbf{D}_k \mathbf{C}_2 \mathbf{Y} + \mathbf{N} \mathbf{B}_k \mathbf{C}_2 \mathbf{Y}\\
& + \mathbf{X} \mathbf{B}_2 \mathbf{C}_k \mathbf{M}^{\text{T}} + \mathbf{N} \mathbf{A}_k \mathbf{M}^{\text{T}},\\[0.2cm]
\widehat{A}_d&= \mathbf{X} \mathbf{A}_d \mathbf{Y} + \mathbf{X} \mathbf{B}_2 \mathbf{D}_k \mathbf{C}_{2d} \mathbf{Y} + \mathbf{N} \mathbf{B}_k \mathbf{C}_{2d} \mathbf{Y}\\
& + \mathbf{X} \mathbf{B}_2 \mathbf{C}_{dk} \mathbf{M}^{\text{T}} + \mathbf{N} \mathbf{A}_{dk} \mathbf{M}^{\text{T}},\\[0.2cm]
\widehat{B}&= \mathbf{X} \mathbf{B}_2 \mathbf{D}_k + \mathbf{N} \mathbf{B}_k,\\
\widehat{C}&= \mathbf{D}_k \mathbf{C}_2 \mathbf{Y} + \mathbf{C}_k \mathbf{M}^{\text{T}},\\
\widehat{C}_d&= \mathbf{D}_k \mathbf{C}_{2d} \mathbf{Y} + \mathbf{C}_{dk} \mathbf{M}^{\text{T}}.
\end{array}
\end{equation}

\noindent 3. Finally, the controller matrices are computed in the following order:
\begin{equation}
\begin{array}{cl}
\mathbf{C}_{dk}&= (\widehat{C}_d - \mathbf{D}_k \mathbf{C}_{2d} \mathbf{Y}) \mathbf{M} ^{-\text{T}},\\[0.20cm]
\mathbf{C}_k&= (\widehat{C} - \mathbf{D}_k \mathbf{C}_{2} \mathbf{Y}) \mathbf{M} ^{-\text{T}},\\[0.20cm]
\mathbf{B}_k&= \mathbf{N}^{-1} (\widehat{B} -  \mathbf{X} \mathbf{B}_2 \mathbf{D}_k),\\[0.20cm]
\mathbf{A}_{dk}&= -\mathbf{N}^{-1} (\mathbf{X} \mathbf{A}_d \mathbf{Y} + \mathbf{X} \mathbf{B}_2 \mathbf{D}_k \mathbf{C}_{2d} \mathbf{Y} + \mathbf{N} \mathbf{B}_k \mathbf{C}_{2d} \mathbf{Y}\\
&+  \mathbf{X} \mathbf{B}_2 \mathbf{C}_{dk} \mathbf{M} ^{\text{T}} - \widehat{A}_d) \mathbf{M} ^{-\text{T}},\\[0.20cm]
\mathbf{A}_{k}&= -\mathbf{N}^{-1} (\mathbf{X} A \mathbf{Y} + \mathbf{X} \mathbf{B}_2 \mathbf{D}_k \mathbf{C}_{2} \mathbf{Y} + \mathbf{N} \mathbf{B}_k \mathbf{C}_{2} \mathbf{Y}\\
&+ \mathbf{X} \mathbf{B}_2 \mathbf{C}_{k} \mathbf{M} ^{\text{T}} - \widehat{A}) \mathbf{M} ^{-\text{T}}.\\
\end{array}
\end{equation}


The next section examines the application of the proposed robust time-delayed LPV control design method to the MAP regulation problem.

\section{MAP regulation using LPV control}\label{sec:Results}
The MAP dynamic regulation problem is formulated in an LPV framework utilizing the state equations in (\ref{LPV_MAP}) where the state-space matrices are as in (\ref{eq:matrices1}). Moreover, the vector of the target outputs to be controlled is $\mathbf{z}(t) = [\phi \cdot x_e (t) \:\: \psi \cdot u (t)]^{\text{T}}$, \textit{i.e.} \begin{small}
$\mathbf{C}_{1}(\boldsymbol{\rho}(t))=\begin{bmatrix}
0 & 0 & \phi\\
0 & 0 & 0
\end{bmatrix},$
\end{small} $\mathbf{D}_{12}(\boldsymbol{\rho}(t))=[0, \:\: \psi]^{\text{T}}$. The matrices $\mathbf{C}_{1d}(\boldsymbol{\rho}(t))$ and $\mathbf{D}_{11}(\boldsymbol{\rho}(t))$ in (\ref{LPVsystem}) are zero matrices with compatible dimensions. The tracking error which is included in the state $x_e(t)$ and the control effort $u(t)$ are penalized by the weighting scalars $\phi$ and $\psi$, respectively. The choice of the scalars $\phi$, and $\psi$ determines the relative weighting in the optimization scheme and depends on the desired performance objectives. The output-feedback controller is designed to minimize the induced $\mathcal{L}_2$ gain (or $\mathcal{H}_{\infty}$ norm) (\ref{eq:Performance Index}) of the closed-loop LPV system (\ref{eq:closed-loop system}) with the augmented uncertain matrices. The design objective is to guarantee closed-loop stability and minimize the worst case disturbance amplification over the entire range of model parameter variations. Theorem \ref{thm:thm4} is used to design a robust LPV output-feedback controller which leads to an infinite-dimensional convex optimization problem with an infinite number of LMIs and decision variables. To overcome this challenge, we utilize the gridding approach introduced in \cite{apkarian1998advanced} to convert the infinite-dimensional problem to a finite-dimensional convex optimization problem. In this regard, we choose the functional dependence as $\mathbf{M}(\boldsymbol{\rho}(t))=\mathbf{M}_0 + \sum\limits_{i=1}^{s}\rho_i(t) \mathbf{M}_{i_1} + \frac{1}{2}\sum\limits_{i=1}^{s}\rho_i^2(t) \mathbf{M}_{i_2}$, where $\mathbf{M}(\boldsymbol{\rho}(t))$ represents any of the parameter-dependent matrices appearing in the LMI condition (\ref{eq:LMI closed-loop}). Finally, gridding the scheduling parameter space at appropriate intervals leads to a finite set of LMIs to be solved for the unknown matrices and $\gamma$. The MATLAB\textsuperscript{\tiny\textregistered} toolbox YALMIP is used to solve the introduced optimization problem \cite{lofberg2004yalmip}.

To evaluate the performance of the proposed robust LPV gain-scheduling output-feedback control design, collected animal experiment data is used to build a  patient's non-linear MAP response model based on (\ref{eq:MAP response TF}) where the instantaneous values of the model parameters $K(t)$, $T(t)$, and $\tau (t)$ are generated as follows \cite{Luspay2015}.

\begin{itemize}
\item Sensitivity parameter, $K(t)$: experiments have demonstrated a regressive non-linear relationship between the vasoactive drug injection and the MAP response through which the patient's sensitivity decreases gradually on a constant vasoactive drug injection. This behavior is captured by the following non-linear relationship:
\begin{equation}
a_k \dot{K}(t) + K (t) = k_0 exp\{ - k_1 i(t) \},
\label{eq:sensitivity}
\end{equation}
where $i(t)$ is the drug injection and $a_k$, $k_0$, and $k_1$ are uniformly distributed random coefficients based on Table \ref{tab:coef} \cite{Craig2004, Flancbaum1997}. For example, a non-responsive patient to the injected vasoactive drug will be characterized by a low $k_0$ and a high $k_1$. 

\item Lag time, $T(t)$: This parameter gradually increases with the injected drug volume and it can be modeled as:
\begin{equation}
T(t)= sat_{\:[T_{\min}, T_{\max}]} \: \{b_{T} \int_{0}^{t} i(t) \:dt \},
\label{eq:lag}
\end{equation}
where $b_{T}$ is a uniformly distributed random variable which shows the inclination of the increase and varies as shown in Table \ref{tab:coef}.

\item  Injection delay, $\tau (t)$: Based on observations, the delay value has a peak shortly after the drug injection and it decays afterward. The following equation is used to describe the delay behavior:
\begin{equation}
\begin{cases}a_{\tau,2} \dddot{\tau}(t) + a_{\tau,1} \ddot{\tau}(t) + \dot{\tau}(t) = b_{\tau,1} \dot{i}(t) + i(t), & \:\:\:\:\:\:\: t\geq t_{i_0}, \\\tau(t)=0, & otherwise,\end{cases}
\label{eq:delay}
\end{equation}

\noindent where the saturation is imposed on the delay value, \textit{i.e.} $sat_{[\tau_{\min}, \tau_{\max}]} \: \tau$ and the uniformly distributed random variables $a_{\tau,2}$, $a_{\tau,1}$, and $b_{\tau,1}$ are listed in Table \ref{tab:coef}.
\end{itemize}

\begin{table}[!t]
\processtable{Probabilistic distribution of the non-linear patient coefficients\label{tab:coef}}
{\begin{tabular*}{15pc}{@{\extracolsep{\fill}}ccc@{}}\toprule
& Parameter    & Distribution\\\midrule
& $a_k$        & $\mathcal{U}(500, 600)$                  \\
& $k_0$        & $\mathcal{U}(0.1, 1)$                    \\
& $k_1$        & $\mathcal{U}(0.002, 0.007)$              \\
& $b_{T}$      & $\mathcal{U}(10^{-4}, 3 \times 10^{-4})$ \\
& $a_{\tau,1}$ & $\mathcal{U}(5, 15)$                     \\
& $a_{\tau,2}$ & $\mathcal{U}(5, 15)$                     \\
& $b_{\tau,1}$ & $\mathcal{U}(80, 120)$                   \\\botrule
\end{tabular*}}{}
\end{table}
The non-linear patient simulation model developed following the above scheme is utilized along with the real-time model parameter estimation provided by the MMSRCKF to validate the proposed LPV control in closed-loop simulations. The MMSRCKF estimates the model parameters of the non-linear patient online and feeds them to the LPV controller as the scheduling parameters as shown in Fig. \ref{fig:system structure}.  

			For comparison purposes, we evaluate the proposed controller performance against a fixed structure PI controller (see \cite{wassar2014automatic}). Given the following nominal values of the model parameters $\overline{K} = 0.55, \overline{T}=150$, and $\overline{\tau}=40$, the tuned PI controller transfer function is as follows:	
        \begin{equation}
	     G_{c} (s) = 3 + \frac{0.017}{s},
		\end{equation}	    
	   \noindent which is obtained based on the approach proposed in \cite{Zhong2006} to meet prescribed gain and phase margin constraints. In the absence of disturbances and measurement noise, the tracking profile and the control effort are shown in Fig. \ref{fig:fig_tracking1} where the objective is to regulate the MAP response to track the commanded MAP with minimum overshoot and settling time and zero steady-state error. According to this figure, the overshoot of the closed-loop response remains within the admissible range and the delay-dependent parameter varying controller provides a faster response with less settling time compared to the conventional PI controller. Next, we assume that the closed-loop system is experiencing both measurement noise and output disturbances. These disturbances could be the result of medical interventions and physiological variations due to hemorrhage or other medications like lactated ringers (LR). Fig. \ref{fig:Disturbance} is a typical profile of such disturbances. Considering measurement noise as white noise with the intensity of $10^{-3}$ the performance of the LPV and the PI controllers can be seen in Fig. \ref{fig:trackingdistnoise}. As expected, the proposed LPV controller outperforms the fixed structure PI controller with respect to rise time and speed of the response due to its scheduling structure.
			
			To evaluate the robustness of the proposed LPV control design, the closed-loop response in the presence of parameter uncertainty on the model parameters is investigated. To this end, the time-delay and the sensitivity are considered to be under-estimated by $30 \%$ and the time constant is considered to be over-estimated by $30 \%$ to result in a worst-case perturbation scenario. The closed-loop MAP response of the system with the proposed robust LPV control design is compared to the response of the LPV controller designed without considering uncertainty obtained using the results of Theorem \ref{thm:thm3}. As shown in Fig. \ref{fig:trackingunc}, the control without considering uncertainty in the design demonstrates oscillatory behavior and higher overshoot both in the closed-loop MAP response and also in the PHP injection which are undesirable. As the results suggest, the proposed robust LPV control design is capable of compensating for the parameter uncertainty.
	
We conclude that the proposed MMSRCKF online parameter estimation method and the proposed LPV gain-scheduling control methodology demonstrates desirable closed-loop performance in terms of commanded MAP tracking and disturbance rejection under different scenarios in the presence of model parameter variations, varying time-delay, model uncertainty and disturbances.

\begin{figure}[!t] 
\hspace*{-.12in}
\centering \includegraphics[width=\columnwidth, height=2.5in]{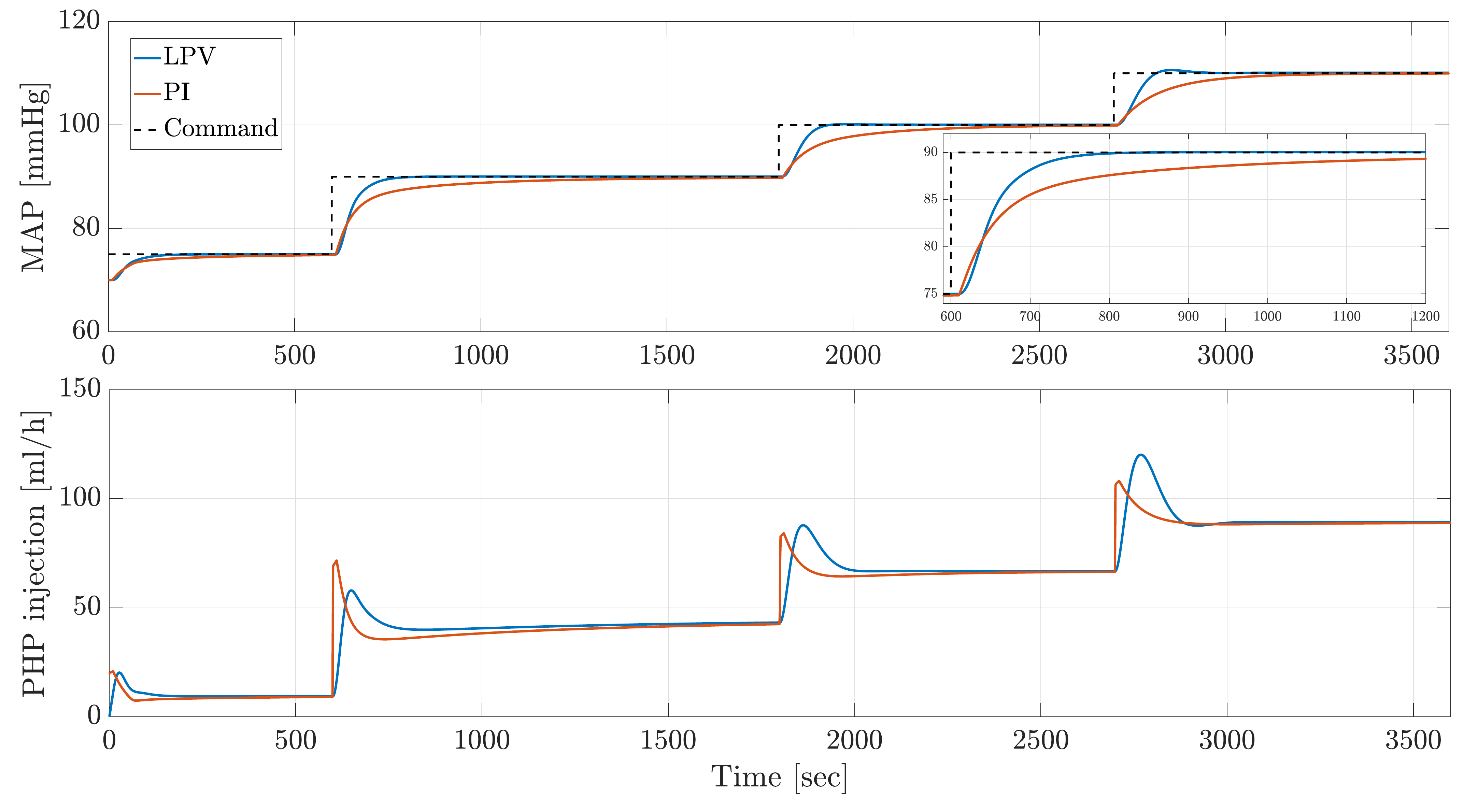}
\caption{Closed-loop MAP response and control effort of the LPV controller and the fixed structure PI controller with no disturbance and no measurement noise} 
\label{fig:fig_tracking1}
\end{figure}

\begin{figure}[!t] 
\hspace*{-.12in}
\centering \includegraphics[width=\columnwidth, height=2.0in]{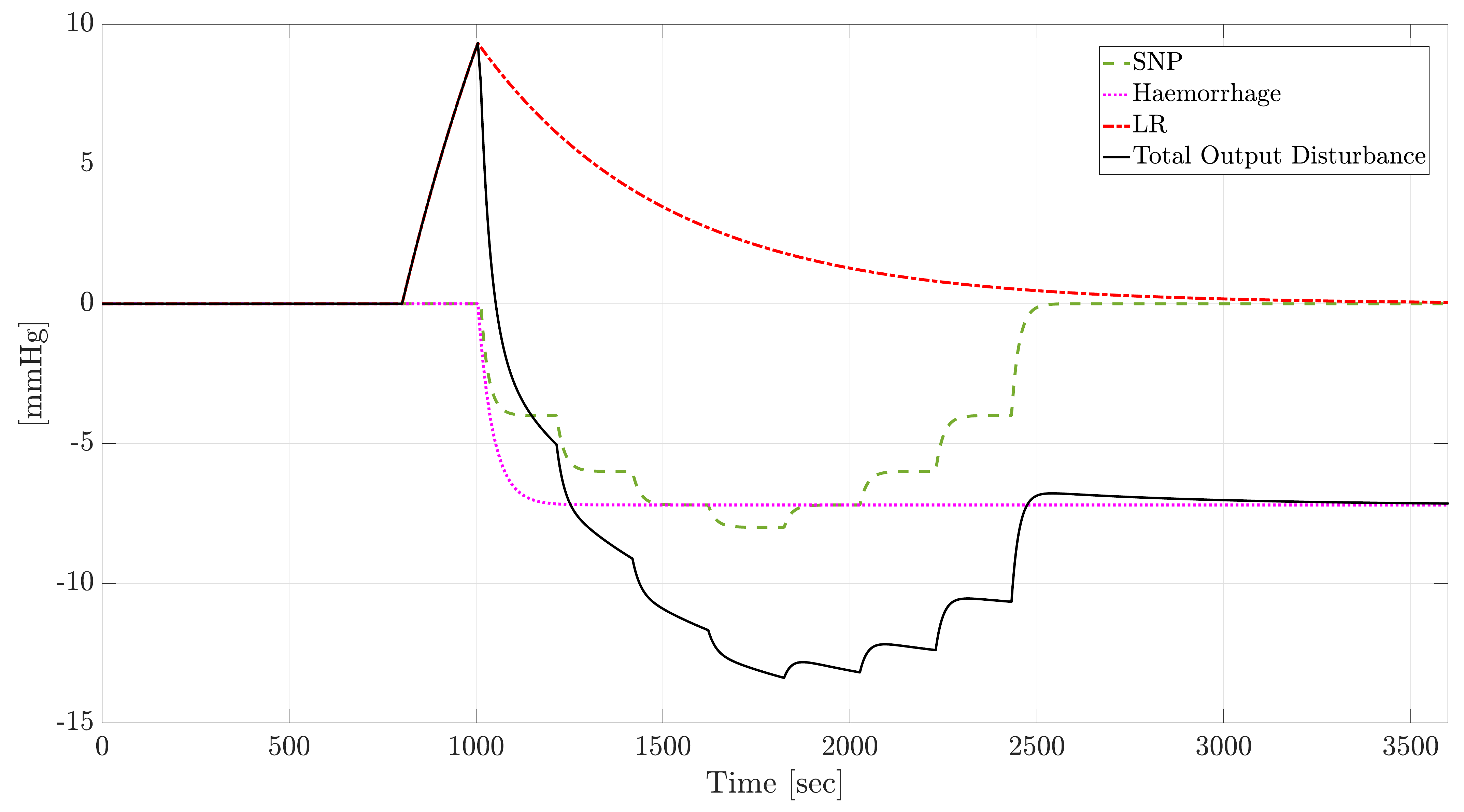}
\caption{Profile of output disturbances} 
\label{fig:Disturbance}
\end{figure}

\section{Conclusion}\label{sec:Conclusion}
Parameter estimation of a MAP dynamic model in response to  vasopressor drug infusion has been examined using a multiple-model square root cubature Kalman filtering algorithm. A first-order dynamic model with time-varying parameters and a time-varying delay is used to capture the MAP variation characteristics. The multiple-model part of the filter accomplishes the delay estimation while the Bayesian-based SRCKF part estimates the remaining four parameters, namely the sensitivity, lag-time, MAP variation as well as its baseline value at each time step using the nonlinear dynamic model. The convergence of the filter is guaranteed by considering the residuals to be zero-mean white noise and the results verify the effectiveness of this approach in comparison to experimental data. The proposed estimation is utilized in conjunction with a feedback control of drug infusion for automated MAP regulation. To this end, the design of a robust LPV output-feedback controller is addressed to track a target MAP profile in the face of model uncertainties, a varying time delay, clinically induced disturbances, and noise. Sufficient conditions for stabilization and disturbance rejection are obtained via bounding the derivative of a proposed Lyapunov-Krasovskii functional and the results are formulated in a parameter-dependent LMI setting. A nonlinear simulation model constructed using animal experiment data is used to validate the closed-loop response of the proposed robust LPV controller in regulating MAP to a target value in comparison with a fixed structure PI controller.

\begin{figure}[!t]
        \includegraphics[width=\columnwidth, height=2.5in]{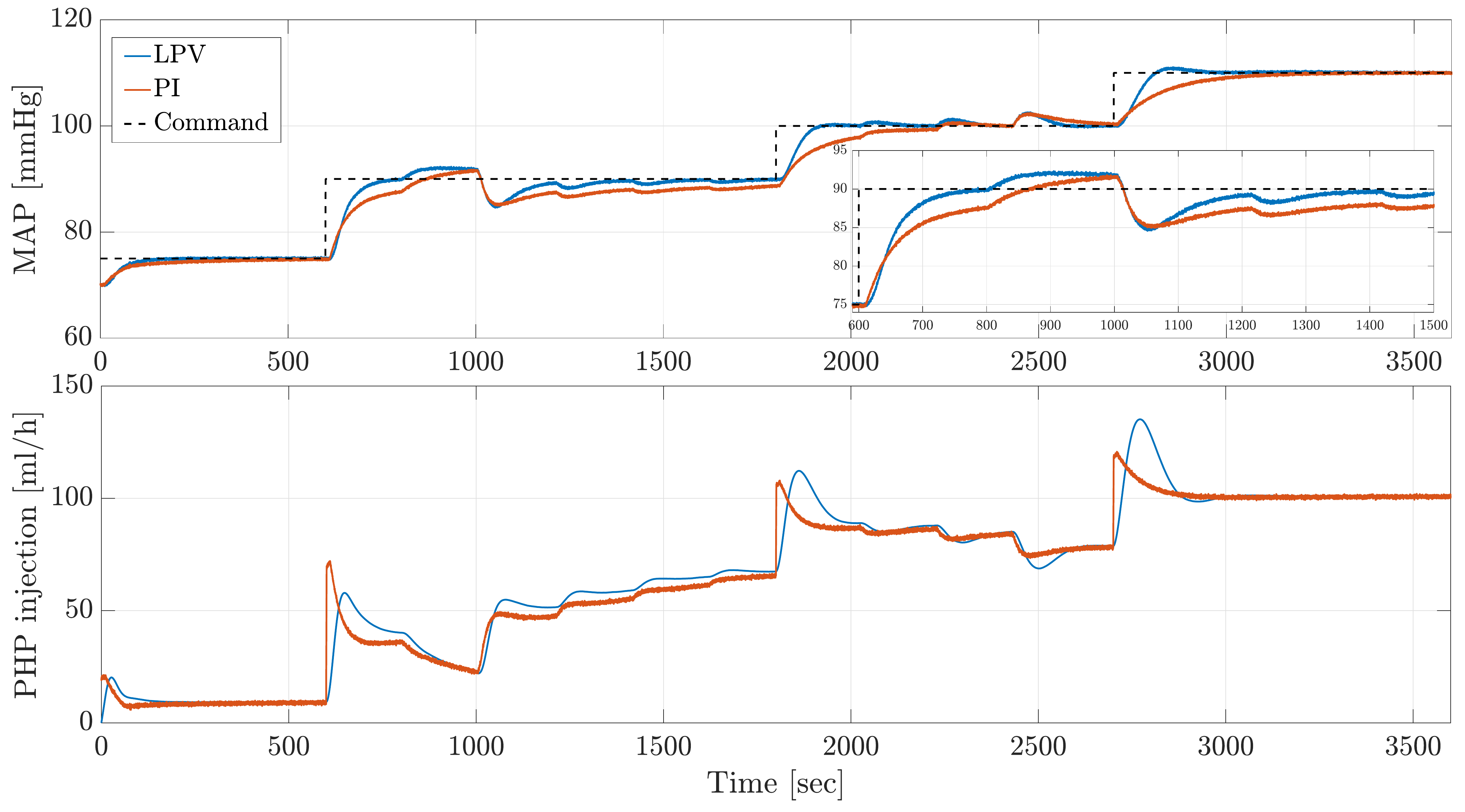}
    \caption{Closed-loop MAP response and control effort of LPV controller against fixed structure PI controller subject to disturbance and measurement noise}
            \label{fig:trackingdistnoise}

\end{figure}

\begin{figure}[!t]
        \includegraphics[width=\columnwidth, height=2.5in]{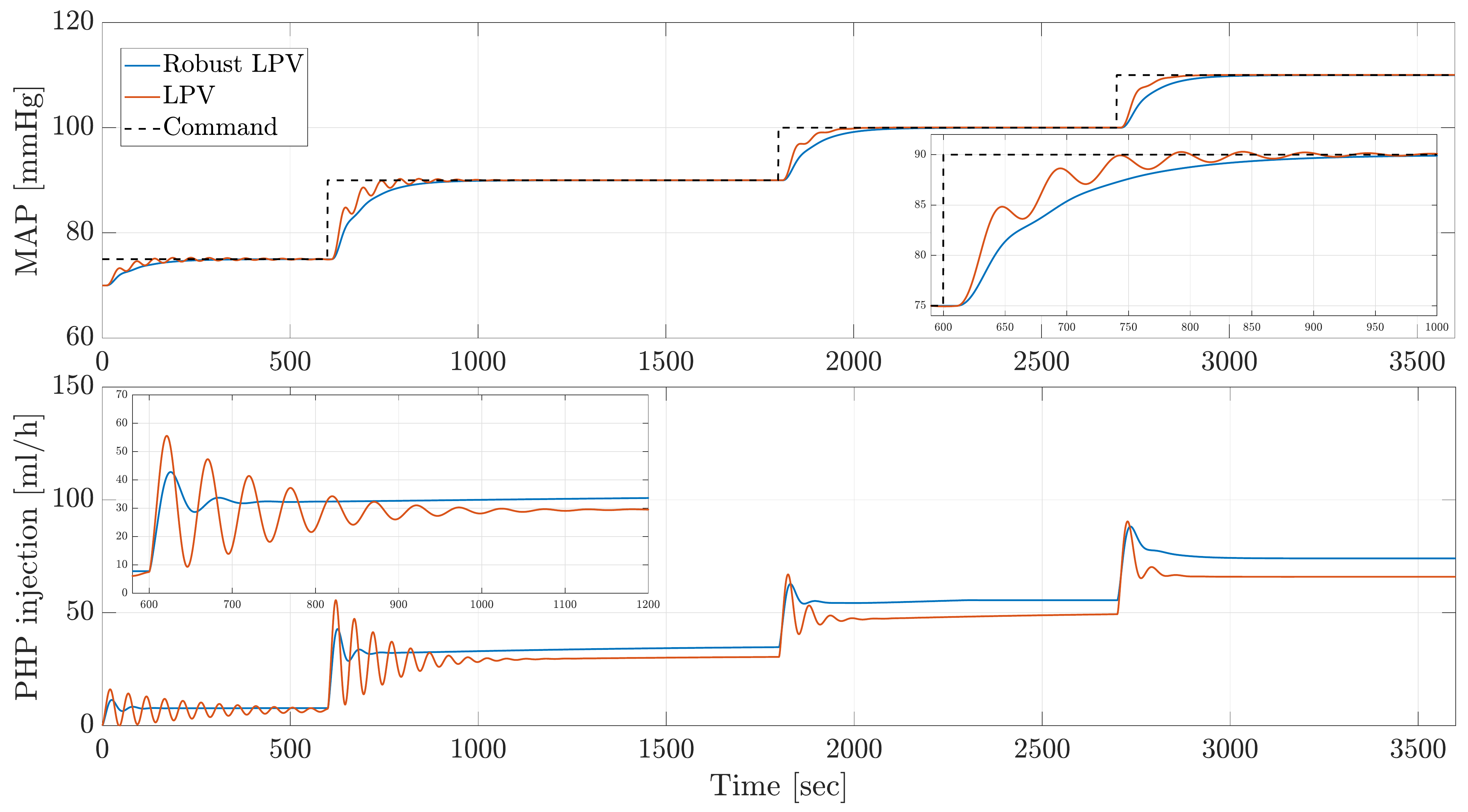}
    \caption{Closed-loop MAP response and control effort of robust LPV controller in the presence of model parameter uncertainty}
            \label{fig:trackingunc}

\end{figure}	

\section*{Acknowledgement}\label{sec:Acknow}
Financial support from the National Science Foundation under grant CMMI1437532 is gratefully acknowledged. The collaboration of the Resuscitation Research Laboratory (Dr. G. Kramer) at the University of Texas Medical Branch (UTMB), Galveston, Texas, in providing animal experiment data is gratefully acknowledged.


\end{document}